\documentclass[english,aps,twocolumn]{revtex4}
\usepackage[T1]{fontenc}
\usepackage[latin9]{inputenc}
\usepackage{amsmath}
\usepackage{amsthm}
\usepackage{amssymb}
\usepackage{graphicx}

\makeatletter
\@ifundefined{textcolor}{}
{%
 \definecolor{BLACK}{gray}{0}
 \definecolor{WHITE}{gray}{1}
 \definecolor{RED}{rgb}{1,0,0}
 \definecolor{GREEN}{rgb}{0,1,0}
 \definecolor{BLUE}{rgb}{0,0,1}
 \definecolor{CYAN}{cmyk}{1,0,0,0}
 \definecolor{MAGENTA}{cmyk}{0,1,0,0}
 \definecolor{YELLOW}{cmyk}{0,0,1,0}
}
  \theoremstyle{plain}
  \newtheorem{prop}{\protect\propositionname}

\newcommand{\be }{\begin {equation}} \newcommand{\ee }{\end {equation}}

\newcommand{\ket}[1]{|#1\rangle}
\newcommand{\bra}[1]{\langle #1|}

\usepackage{subfigure}

\usepackage{enumitem}

\usepackage{babel}
\providecommand{\propositionname}{Proposition}

\makeatother

\usepackage{babel}
  \providecommand{\propositionname}{Proposition}

\begin{document}

\title{State independent uncertainty relations from eigenvalue minimization }

\author{Paolo Giorda}
\email{magpaolo16@gmail.com}

\affiliation{Dip. Fisica and INFN Sez. Pavia, University of Pavia, via Bassi 6,
I-27100 Pavia, Italy; Consorzio Nazionale Interuniversitario per la
Scienze fisiche della Materia (CNISM), Italy;}

\author{Lorenzo Maccone}

\affiliation{Dip. Fisica and INFN Sez. Pavia, University of Pavia, via Bassi 6,
I-27100 Pavia, Italy.}

\author{Alberto Riccardi}

\affiliation{Dip. Fisica and INFN Sez. Pavia, University of Pavia, via Bassi 6,
I-27100 Pavia, Italy.}
\begin{abstract}
We consider uncertainty relations that give lower bounds to the sum
of variances. Finding such lower bounds is typically complicated,
and efficient procedures are known only for a handful of cases. In
this paper we present procedures based on finding the ground state
of appropriate Hamiltonian operators, which can make use of the many
known techniques developed to this aim. To demonstrate the simplicity
of the method we analyze multiple instances, both previously known
and novel, that involve two or more observables, both bounded and
unbounded. 
\end{abstract}
\maketitle

\section{Introduction}

Preparation uncertainty relations capture the essence of quantum mechanics:
not all properties of a quantum system can be exactly defined at once
\cite{URHistory,URRobertson,WheelerZurek}. While quantum complementarity
tells us that there exist complementary properties which can be assigned
to a system, but that cannot have joint definite values, uncertainty
relations go even beyond this very counterintuitive concept: they
tell us that complementary properties can be defined at least partially,
as long as we do not require them to be determined with perfect precision.
The uncertainty relations then are doubly counterintuitive: they originate
from complementarity, but then, in a sense, allow to partially counterbalance
the effects of complementarity. In addition to the foundational issues
\cite{busch,hall,lahti}, uncertainty relations have found applications
in a variety of problems such as entanglement detection \cite{URHoffmanEntanglementDetection,URGuneEntanglementDetection},
spin squeezing \cite{NoriSpinSqueezReview}, quantum metrology \cite{HePQS}.
The conventional treatment of preparation uncertainties follows the
Heisenberg-Robertson approach \cite{URRobertson} which involves the
product of uncertainties, in order to employ the Cauchy-Schwartz inequality
in their derivation. They are expressed in terms of variances of incompatible
observables e.g. $\Delta^{2}A\Delta^{2}B\ge\left|\bra{\psi}\left[A,B\right]\ket{\psi}\right|$
for observables $A$ and $B$. However, the lower bound for product
of variances may be null for some state $\ket{\psi}$, and thus non-informative.
Or it is null whenever one of the two variances is i.e., when $\ket{\psi}$
is a (proper) eigenstates of one of the observables. This prevents
the interpretation of the product uncertainty relations as a true
measure of how incompatible are two observables, where we assume that
observables are compatible if their value can be precisely jointly
assigned for at least one state of the system. For these reasons,
it is preferable to consider uncertainty relations that give a lower
bound to the sum of variances $\Delta^{2}A+\Delta^{2}B$ of two or
more operators \cite{SURRivas,SURMacconePati,SURXiao,SURUnitaryPati,SURExperimMa,SURSU2Bjork,SURNObsChen,SURNObsSong,SURExperimChen}.
Furthermore, the case of two observables has important physical applications,
for example in quantum metrology protocols where the squeezing of
two angular momentum operators (planar quantum squeezing) allows for
phase uncertainties below the standard quantum limit \cite{HePQS,PQSNDMeasMitchell,HePQSEntangInterfero,PQSExpermMitchell},
or in quantum information strategies for detecting entanglement \cite{URHoffmanEntanglementDetection,WernerNumericalRange}.
In this paper, we present a general procedure to derive a state independent
lower bound for the sum of variances of an arbitrary number $N$ of
Hermitian operators $A_{n}$

\begin{eqnarray*}
V_{Tot}\left(\ket{\psi}\right) & = & \sum_{n=1}^{N}\Delta_{\ket{\psi}}^{2}A_{n}
\end{eqnarray*}
where the variances are calculated on an arbitrary state $|\psi\rangle$.
The largest possible value of $l_B$ that satisfies $V_{Tot}\left(\ket{\psi}\right)\ge l_{B}$
, ideally one that satisfies it with equality for some state, constitutes
the best attainable lower bound that depends only on the observables.
In contrast to previous derivations, our method is based on the search
of the ground states energy $\varepsilon_{gs}$ of specifically designed
Hamiltonian operators, and can use the multitude of techniques developed
to this aim. In general this allows to easily and quickly find good
approximations $\tilde{l}_{B}$ of $l_B$. \\The strategies proposed
to date for determining $l_{B}$ are based on different approaches.
In \cite{SIndepURMaccone} the Authors have devised a method to (analytically)
identify $l_{B}$, provided the operators $A_{n}$ are the generators
of a Lie algebra. In \cite{abbot} the case of arbitrary qubit observables
is considered. Other methods are focused in finding $l_{B}$ or at
least a sufficiently good approximations $\tilde{l}_{B}$ that may
or may not be achievable; they fall in two different classes: the
strategies ``from above'' and the strategies ``from below''. The
former are based on algorithms that find $l_{B}$ by, possibly iteratively,
starting from approximations $\tilde{l}_{B}^{+}\ge l_{B}$. The most
obvious of such strategies use numerical minimization algorithms that
scan the whole $M$ dimensional Hilbert space $\mathcal{H}_{M}$ of
the system searching for $l_{B}$. Since the procedure requires the
identification of the $2M-2$ real coefficients of the state $\ket{\psi_{min}}$
which minimizes the sum of variances, it is numerically demanding
when $M$ is large and is prone to errors due to the possibility of
getting trapped in some local minima. A sophisticated procedure ``from
above'' has been put forward in \cite{WerAngularMomentum}, where
a \textit{seesaw} numerical algorithm was devised and used for example
for sum of variances involving angular momentum components. In principle
the algorithm can be used with an arbitrary number $N$ of observables,
and it is based on a alternating minimization procedure which at each
step $i$ determines an approximation $\tilde{l}_{B,i}^{+}\ge l_{B}$.
As the Authors suggests in \cite{WernerNumericalRange}, the strategy
may get trapped in local minima and the proof of its convergence to
the global minimum $l_{B}$ is an open problem. The strategy ``from
below'' is instead based on a elegant mapping of the minimization
problem into a geometric one (joint numerical range), where one searches
for a sequence of polyhedral approximations of a suitable convex set
\cite{WernerNumericalRange,ZycowskyNumricalRange,Szyma=000144skiNumericalRange}.
While in certain simple cases the exact $l_{B}$ can be identified
\cite{ZycowskyNumricalRange}, in other cases at each step $i$ the
algorithm provides both a valid approximation of $l_{B}\ge\tilde{l}_{B,i}^{-}$
from below and an approximation $\tilde{l}_{B,i}^{+}\ge l_{B}$ from
above, such that one has the control over the precision $\epsilon_{i}=\tilde{l}_{B,i}^{+}-\tilde{l}_{B,i}^{-}$
with which the optimal bound $l_{B}$ is approximated \cite{WernerNumericalRange}
. The method has been up to now applied to the sum of variances of
two operators; its generalization to a larger number of operators
requires further geometrical and numerical refinements \cite{WernerNumericalRange}. 

Here we propose a minimization method ``from below'', on the basis
of which one can subsequently also find an approximation ``from above''.
It is based on connecting the sum of variances to an Hamiltonian \footnote{We use ``Hamiltonian'' to indicate an operator whose spectrum is
lower bounded. The Hamiltonian operators we consider in the paper
are not necessarily connected to an energy observable.} expectation value. Whence the search for a bound from below $\tilde{l}_{B}^{-}$
can be mapped onto a search for the Hamiltonian minimum energy. To
begin with the Hamiltonian to minimize is the sum of the operators
defined on an extended Hilbert space $\mathcal{H}_{M}\otimes\mathcal{H}_{M}$
(where $\mathcal{H}_{M}$ is the system space) defined as
\begin{eqnarray}
H_{n} & = & \frac{A_{n}^{2}\otimes\mathbb{I}+\mathbb{I}\otimes A_{n}^{2}}{2}-A_{n}\otimes A_{n}.\label{eq: Hn Definition}
\end{eqnarray}
 Indeed, 
\begin{eqnarray}
V_{Tot}\left(|\psi\rangle\right) & = & \bra{\psi}\bra{\psi}\sum_{n=1}^{N}H_{n}\ket{\psi}\ket{\psi}\label{eq: Mapping VTot HTot}
\end{eqnarray}
i.e. the sum of variances can be written as the average value of the
operator $H_{Tot}=\sum_{n}H_{n}$ on the product state $\ket{\psi}\ket{\psi}\in\mathcal{H}_{M}\otimes\mathcal{H}_{M}$.
Then, the search of a lower bound to the sum of variances maps directly
to the search of the ground state of the total Hamiltonian $H_{Tot}$.
In general, the ground state will not be a factorized state $\ket{\varepsilon_{gs}}\neq \ket{\psi}\ket{\psi}$,
nonetheless, the corresponding ground state energy $\varepsilon_{gs}$
will provide a non-achievable but valid state independent lower bound
to the sum of variances. As we will show, while the mapping (\ref{eq: Mapping VTot HTot})
itself can in certain cases provide the optimal value $l_{B}$ or
close approximations from below $\tilde{l}_{B}^{-}$, it is also the
starting point for devising procedures that give better bounds when
needed. This is especially important since the ground state energy
of $H_{Tot}$ may be null. In this case on one hand we will give a
bound that involves $H_{Tot}$'s first excited state. On the other
hand, we show how by using appropriate modifications $H_{Tot,n}$
of $H_{Tot}$ one can obtain refined approximations of $l_{B}$ from
below in terms of their ground state energies. \\The knowledge of
the ground state of $H_{Tot}$, or of its modifications, via its Schmidt
decomposition allows one to identify a state $\ket{\psi_{sat}}\in\mathcal{H}_{M}$
that provides an approximation ``from above'' i.e., $V_{Tot}\left(\ket{\psi_{sat}}\right)=\tilde{l}_{B}^{+}\ge l_{B}$
. This procedure can always be applied, and the unknown tight bound
$l_{B}$ for the variance sum lies in the interval between the bound
``from above'' $V_{Tot}\left(\ket{\psi_{sat}}\right)$ and the one
``from below'' $\varepsilon_{gs}$. The width of this interval $V_{Tot}\left(\ket{\psi_{sat}}\right)-\varepsilon_{gs}\ge0$
thus provides an indication of the accuracy of the approximations
found, namely how far is the tight bound from the ones obtained. \\We
illustrate our methods using some examples: we analyze both known
cases and derive new uncertainty relations. The known cases show that
our method is able to recover known results easily. And the new results
show that our method can allow to tackle situations difficult to analyze,
such as the case of more than two observables and the infinite dimensional
case for unbounded operators. For each example $i)$ we identify the
relevant operator; $ii)$ we evaluate the relative $\varepsilon_{gs}$,
$\ket{\varepsilon_{gs}},\ket{\psi_{sat}}$; $iii)$ we give the width
of the interval $V_{Tot}\left(\ket{\psi_{sat}}\right)-\varepsilon_{gs}\ge0$.\\In
Sec. \ref{sec: General-Results} we present the first main general
results that one can obtain by mapping the sum uncertainty relations
to a Hamiltonian ground state search. Then in Sec. \ref{sec: Examples}
we apply these results to some examples to demonstrate the versatility
of the method. In particular, in Sec. \ref{subsec: Generators-of-su(2)}
we analyze the uncertainty relations for all the $su(2)$ generators;
in Sec. \ref{subsec:Spin-operators-and} we consider a subset of the
previous operators, namely the planar spin squeezing; in Sec. \ref{subsec:su(3)-operators}
we consider a lower bound for a set of different numbers of operators
chosen from the generators of the $su(3)$ algebra to show how our
method can easily deal with more than two observables; and finally
in Sec. \ref{subsec:Harmonic-oscillator-operators} we analyze the
sum uncertainty relations for one quadrature and the number operator
of a harmonic oscillator, to show that our method can be also applied
to unbounded operators. Some of these examples have already appeared
in the literature, while others refer to novel sum uncertainty relations.
Finally, the appendices contain some technical results and supporting
material.

\section{General Results\label{sec: General-Results}}

\subsection{Properties of the Hamiltonian $H_{Tot}$}

We start by studying the properties of the Hamiltonian $H_{Tot}$,
in particular of its ground state energy $\varepsilon_{gs}$ and ground
state $\ket{\varepsilon_{gs}}$. The discussion will allow us one
hand to describe how $H_{Tot}$ can used to derive the desired lower
bounds, and on the other hand to prepare the ground for the following
developements. As a general premise we choose to base the following
discussions and results on the use of operators $A_{n}$ with non-degenerate
spectrum. This choice allows in the first place to simplify the notations.
While some of the results obtained can be easily extended to the non-degenerate
instances, the latter should be treated on a case by case basis. Furthermore,
we will treat only set of operators with no common eigenstates, otherwise
the problem trivially reduces to having $V_{Tot}=0$. \\With this
setting in mind, we first notice that each operator $H_{n}$ is by
construction semi-definite positive, as it can be seen by writing
it in its diagonal form 
\begin{eqnarray}
H_{n} & = & \frac{1}{2}\sum_{i,j=1}^{M}\left(a_{n,i}-a_{n,j}\right)^{2}\ket{a_{n,i}}\ket{a_{n,j}}\bra{a_{n,i}}\bra{a_{n,j}}\label{eq: Hn diagonal}
\end{eqnarray}
where $\left\{ \ket{a_{n,i}}\right\} $ is the $A_{n}$ eigenbasis
and $\left\{ a_{n,i}\right\} _{i=1}^{M}$ the corresponding eigenvalues,
that by convention in the paper we suppose listed in increasing order.
In particular $H_{n}$ has $\varepsilon_{gs}^{n}=0$ as ground state
energy. The main properties of $H_{Tot}$ are described with the following 
\begin{prop}
\label{Prop 1}Given $N$ Hermitian operators $\left\{ A_{n}\right\} _{n=1}^{N}$
with no common eigen-states, each with non-degenerate eigenspectrum
and eigenbasis $\left\{ \ket{a_{n,i}}\right\} _{i=1}^{M}$, then 

i) if the Hamiltonian $H_{Tot}=\sum_{n}H_{n}$ with $H_{n}$ as in
(\ref{eq: Hn Definition}) has positive ground state energy zero $\varepsilon_{gs}>0$
then 
\begin{eqnarray*}
V_{Tot}\left(\ket{\psi}\right) & \ge & \varepsilon_{gs}
\end{eqnarray*}
 ii) if $ $$\varepsilon_{gs}=0$, then $H_{Tot}$ has a unique ground
state that can be written in any of the eigenbasis $\left\{ \ket{\tilde{a}_{n,i}}\ket{\tilde{a}_{n,i}}\right\} _{i=1}^{M}$
as the maximally entangled state 
\begin{eqnarray}
\ket{\varepsilon_{gs}} & = & \frac{1}{\sqrt{M}}\sum_{i}\ket{\tilde{a}_{n,i}}\ket{\tilde{a}_{n,i}}\label{eq: HTot unique ground state}
\end{eqnarray}
with $\ket{\tilde{a}_{n,i}}=\exp\left(i\phi_{i,n}/2\right)\ket{a_{n,i}},\ \forall n,i$
and $\phi_{n,i}$ appropriate phases. Furthermore, given $\varepsilon_{1}>0$
i.e., the first excited energy of $H_{Tot}$ then 
\begin{eqnarray*}
V_{Tot}\left(\ket{\psi}\right) & \ge & \varepsilon_{1}\left(1-\frac{1}{M}\right)
\end{eqnarray*}
\end{prop}
The Proof of result $i)$ naturally follows from our starting point
(\ref{eq: Mapping VTot HTot}) and the fact that for any $\ket{\psi}\in\mathcal{H}_{M}$
\begin{eqnarray*}
\bra{\psi}\bra{\psi}H_{Tot}\ket{\psi}\ket{\psi} & \ge & \varepsilon_{gs}
\end{eqnarray*}
The Proof of result $ii)$ can be found in the Appendix \ref{sec: Appendix Properties of H_Tot}.
Results $i)$ and $ii)$ show that the mapping introduced in (\ref{eq: Mapping VTot HTot})
has as first consequence the possibility of deriving a non-trivial,
in the sense of non-zero, lower bound for $V_{Tot}\left(\ket{\psi}\right)$
starting from the Hamiltonian $H_{Tot}$. While we do not have general
results that allow to establish in the most general case whether the
ground state energy $\varepsilon_{gs}$ of $H_{Tot}$ is zero or not,
the proposition takes into account both cases.\\ How tight are the
bounds described in Proposition \ref{Prop 1} depends on the problem
at hand. As we shall see in the example (\ref{subsec: Generators-of-su(2)})
$\varepsilon_{gs}\neq0$ and it coincides with the optimal bound $l_{B}$.
On the contrary in the other examples $\varepsilon_{gs}\neq0$ and/or
$\varepsilon_{1}\left(1-\frac{1}{M}\right)$ represent a meaningful
approximation $\tilde{l}_{B}^{-}$ of $l_{B}$ when the dimension
$M$ of the underlying Hilbert space is small; while for large $M$
these values may be far from the actual $l_{B}$, for example they
do not grow with $M$. To cope with these situations, and derive state
independent lower bounds that are closer to the optimal one $l_{B}$,
we provide different strategies that are based on modified versions
of $H_{Tot}$. 

\subsection{State independent lower bounds from modifications of $H_{Tot}$}

We illustrate the strategies in two steps. We start with Proposition
\ref{Prop-2} and derive a lower bound for the set of states that
have null expectation value for at least one of the operators $A_{n}$.
The method that will allow to include all states in $\mathcal{H_{M}}$
will be described in Proposition \ref{Prop-3} as an extension of
the following result 
\begin{prop}
\label{Prop-2}Given the Hamiltonian $H_{Tot}$, then for each $n$
the Hamiltonian
\begin{eqnarray*}
H_{Tot,n} & = & H_{Tot}+A_{n}\otimes A_{n}=\\
 & = & \sum_{m\neq n}H_{n}+\frac{A_{n}^{2}\otimes\mathbb{I}+\mathbb{I}\otimes A_{n}^{2}}{2}
\end{eqnarray*}

i) is positive definite;

ii) its ground state energy $\varepsilon_{gs,n}>0$ provides a non-zero
lower bound of $V_{Tot}$ for all the set of states 
\[
S_{n}^{0}=\left\{ \ket{\phi}\in\mathcal{H}_{M}|\left\langle \phi\left|A_{n}\right|\phi\right\rangle =0\right\} ;
\]

iii) the lower bound for the set of states $\cup_{n}S_{n}^{0}\subseteq\mathcal{H}_{M}$
i.e., those states which have null expectation value for at least
one operator $A_{n}$ is given by 
\begin{eqnarray*}
\min_{n} & \varepsilon_{gs,n} & >0
\end{eqnarray*}
\end{prop}
\begin{proof}
To prove result $i)$ we first observe that $H_{Tot,n}$ is obviously
definite positive whenever $A_{n}^{2}$ is. If this is not the case,
since we are dealing with operators with non-degenerate spectrum,
$A_{n}^{2}$ has a unique eigenstate $\ket{a_{n,1}}$ corresponding
to the eigenvalue $a_{n,1}=0$. Due to the structure of each kernels
$Ker(H_{m})$ of the operators $H_{m},\ m\neq n$, equation (\ref{eq: Ker Hn eigenbasis})
in Appendix \ref{sec: Appendix Properties of H_Tot}, the only product
states in any of the $Ker(H_{m})$ have the form $\ket{a_{m,i}}\ket{a_{m,i}}$;
but since we have supposed that the operators $\left\{ A_{n}\right\} _{n=1}^{N}$
have no common eigenstates $\ket{a_{n,1}}\ket{a_{n,1}}\notin Ker(H_{m}),\ m\neq n$;
therefore it must be $H_{Tot,n}>0$. Result $ii)$ follows from the
fact that for all states in $S_{n}^{0}$ 
\begin{eqnarray*}
\bra{\phi}\bra{\phi}H_{Tot}\ket{\phi}\ket{\phi} & = & \bra{\phi}\bra{\phi}H_{Tot,n}\ket{\phi}\ket{\phi}+\\
 & - & \bra{\phi}\bra{\phi}A_{n}\otimes A_{n}\ket{\phi}\ket{\phi}=\\
 & = & \bra{\phi}\bra{\phi}H_{Tot,n}\ket{\phi}\ket{\phi}=\\
 & \ge & \bra{\varepsilon_{gs,n}}H_{Tot,n}\ket{\varepsilon_{gs,n}}
\end{eqnarray*}
One can then determine the following lower bound 
\begin{eqnarray*}
\min_{n} & \varepsilon_{gs,n} & >0
\end{eqnarray*}
for the union $\cup_{n}S_{n}\subseteq\mathcal{H}_{M}$. Indeed, if
$\varepsilon_{gs,n}>\varepsilon_{gs,m},\ n\neq m$ then $\varepsilon_{gs,m}$
is a lower bound for both set of states belonging to $S_{n}$ and
$S_{m}$.
\end{proof}
As we shall see in the following, in specific cases it turns out that
all $\varepsilon_{gs,n}=\tilde{l}_{B}^{-}$ are equal $\forall n$
and, thanks to the symmetries of the problem, finding the ground state
energy of a single Hamiltonian $H_{Tot,n}$ allows to determine the
required lower bound. However, when no such symmetry properties are
available, in general $\cup_{n}S_{n}\subset\mathcal{H}_{M}$ i.e.,
$\cup_{n}S_{n}$ may only be a proper subset of $\mathcal{H}_{M}$,
and the optimization is not sufficient. Therefore a different procedure
must be devised to find a lower bound for all states in $\mathcal{H}_{M}$.
To this aim for fixed $n$ we first define the operator $A_{n}^{\alpha}=A_{n}-\alpha\mathbb{I}$;
then $\forall\alpha\in\left[a_{n,1},a_{n,M}\right]$ one has that
$\Delta^{2}A_{n}^{\alpha}=\Delta^{2}A_{n}$ and one can define the
Hamiltonian 
\begin{eqnarray*}
H_{n}^{\alpha} & = & \frac{\left(A_{n}^{\alpha}\right)^{2}\otimes\mathbb{I}+\mathbb{I}\otimes\left(A_{n}^{\alpha}\right)^{2}}{2}-A_{n}^{\alpha}\otimes A_{n}^{\alpha}
\end{eqnarray*}
and the total Hamiltonian
\begin{eqnarray*}
H_{Tot}^{\alpha} & = & \sum_{m\neq n}H_{m}+H_{n}^{\alpha}
\end{eqnarray*}
Simply by substitution, one can verity that $H_{n}^{\alpha}=H_{n}$
and $H_{Tot}=H_{Tot}^{\alpha}$ Therefore $\forall\alpha\in\left[a_{n,1},a_{n,M}\right]$
if $\ket{\psi_{min}}$ minimizes $V_{Tot}$ then 
\begin{eqnarray*}
V_{Tot}^{min} & = & \bra{\psi_{min}}\bra{\psi_{min}}H_{Tot}^{\alpha}\ket{\psi_{min}}\ket{\psi_{min}}
\end{eqnarray*}
The strategy that allows one to find a state independent lower bound
can now be expressed as follows 
\begin{prop}
\label{Prop-3}For each $n$ and for each $\alpha\in\left[a_{n,1},a_{n,M}\right]$,
define the Hamiltonian $H_{Tot,n}^{\alpha}=H_{Tot}^{\alpha}+A_{n}^{\alpha}\otimes A_{n}^{\alpha}$
with non-zero ground state energy $\varepsilon_{gs,n}^{\alpha}>0$.
Then

i) for fixed $n$ it holds that $\forall\ket{\phi}\in\mathcal{H}_{M}$
\begin{eqnarray*}
V_{Tot}\left(\ket{\phi}\right) & \ge & \min_{\alpha\in\left[a_{n,1},a_{n,M}\right]}\varepsilon_{gs,n}^{\alpha}
\end{eqnarray*}
 and $\min_{\alpha}\varepsilon_{gs,n}^{\alpha}$ provides a state
independent lower-bound;

ii) the best lower bound $\forall\ket{\phi}\in\mathcal{H}_{M}$ is
given by 
\begin{eqnarray*}
\max_{n}\min_{\alpha\in\left[a_{n,1},a_{n,M}\right]}\varepsilon_{gs,n}^{\alpha} & > & 0
\end{eqnarray*}
\end{prop}
\begin{proof}
Since $\left(A_{n,i}^{\alpha}\right)^{2}$ is diagonal in the same
basis of $\left(A_{n,i}^{\alpha=0}\right)^{2}$, the positivity of
$H_{Tot,n}^{\alpha}$ can be demonstrated in the same way it was shown
in Proposition \ref{Prop-2} for $H_{Tot,n}^{\alpha=0}$ . In order
to prove result $i)$ we first define the set $S_{n}^{\alpha}=\left\{ \ket{\phi}\in\mathcal{H}_{M}|\left\langle \phi\left|A_{n}\right|\phi\right\rangle =\alpha\right\} $;
then $\forall\ket{\phi}\in S_{n}^{\alpha}$ $\left\langle \phi\left|A_{n}^{\alpha}\right|\phi\right\rangle =0$
and
\begin{eqnarray*}
V_{Tot}\left(\ket{\phi}\right) & = & \bra{\phi}\bra{\phi}H_{Tot}\ket{\phi}\ket{\phi}=\\
 & = & \bra{\phi}\bra{\phi}H_{Tot,n}^{\alpha}\ket{\phi}\ket{\phi}-\bra{\phi}\bra{\phi}A_{n}^{\alpha}\otimes A_{n}^{\alpha}\ket{\phi}\ket{\phi}=\\
 & = & \bra{\phi}\bra{\phi}H_{Tot,n}^{\alpha}\ket{\phi}\ket{\phi}=\\
 & \ge & \bra{\varepsilon_{gs,n}^{\alpha}}H_{Tot,n}^{\alpha}\ket{\varepsilon_{gs,n}^{\alpha}}=\varepsilon_{gs,n}^{\alpha}
\end{eqnarray*}
For $\alpha$ belonging to the spectrum of $A_{n}$ it holds $\cup_{\alpha\in\left[a_{n,1},a_{n,M}\right]}S_{n}^{\alpha}\equiv\mathcal{H}_{M}$
and one obtains $i)$. Result $ii)$ is therefore a simple consequence
of the fact that, for each $n$, $\min_{\alpha}\varepsilon_{gs,n}^{\alpha}$
is a lower bound for \textit{all states} in $\mathcal{H}_{M}$; and
the maximum of these values gives the highest lower bound obtainable
by means of the above defined Hamiltonians.
\end{proof}
The Propositions \ref{Prop 1}-\ref{Prop-3} constitute the main general
results of our work. They show that the mapping (\ref{eq: Mapping VTot HTot})
allows one to reduce the problem of finding the lower bound for $V_{Tot}$
to an eigenvalue problem. There are at least three different ways
of obtaining the desired lower bound: $a)$ one can work directly
with $H_{Tot}$; $b)$ one can use a single Hamiltonian $H_{Tot,n}^{\alpha}$
for some specific $n$; $c)$ in order to further optimize the result
one can use the $H_{Tot,n}^{\alpha}$ for all $n$. Before passing
to analyze different examples we want first discuss the limits and
virtues of the outlined approach. \\We start with the possible limits.
The procedure is in the first place based on the evaluation of the
ground state energy of Hamiltonians acting on $\mathcal{H}_{M}\otimes\mathcal{H}_{M}$
and thus have dimension $M^{2}\times M^{2}$ that can in principle
be very large. Furthermore, in order to obtain the best result $ii)$
in Proposition \ref{Prop-3} the procedure outlined requires in general
a minimization over $\alpha$ for each $n$, that in principle, e.g.
when the dimension of the Hilbert space $M$ or the number of operatorsn
$N$ is large, and/or the intervals $\left[a_{n,1},a_{n,M}\right]$
are very large, can be numerically demanding. \\As for the virtues,
in the first place the procedure is based on the evaluation of ground
state energies, a task for which very efficient and stable routines
are available, even for large dimensions, especially if the Hamiltonians
have some simple form (e.g. sparse, banded, ect.). Secondly, in order
to obtain a state independent lower bound one in principle only need
to choose one of the Hamiltonians $H_{Tot,n}^{\alpha}$ i.e., choose
a specific $n$, and then only one optimization over $\alpha\in\left[a_{n,1},a_{n,M}\right]$
is needed; for example one could choose $n$ such that the interval
$\left[a_{n,1},a_{n,M}\right]$ is the smallest possible. Furthermore,
one can be interested in a lower bound that, though being strictly
speaking state dependent, is very simple to achieve. For example if
for the physical problem at hand only states with specific average
values are relevant, e.g. states with fixed average $\left\langle \phi\left|A_{n}\right|\phi\right\rangle =\alpha_{fix}$,
the optimization procedure simply requires the evaluation of the single
ground state energy $\varepsilon_{gs,n}^{\alpha_{fix}}$. The procedure
can therefore be flexibly adapted to various specific needs and/or
to obtain partial results.

The above reasonings are valid for the most general case i.e., when
there is no structure in the problem, and the $A_{n}$'s are totally
unrelated. However, as we will show in the following examples, there
may be situations where the presence of some constraints, e.g. symmetries,
allow to drastically reduce the complexity of the problem. This can
be solved by either reducing the problem to an equivalent one which
has known analytic solution, or by evaluating a single ground state
energy, instead of minimizing over $\alpha$. Indeed suppose for example
that $V_{Tot}\left(U\ket{\psi}\right)=V_{Tot}\left(\ket{\psi}\right)$
where $U$ is a unitary operator acting on $\mathcal{H}_{M}$ that
represents a symmetry for $V_{Tot}$. Then one has immediately that
$U^{\dagger}\otimes U^{\dagger}H_{Tot}U\otimes U=H_{Tot}$, such that
the symmetries of $V_{Tot}$ can be translated into symmetries of
$H_{Tot}$ and can be exploited \textit{in the Hamiltonian framework}
with the aim of simplifying the evaluation of the relative lower bounds.
In this respect we now give a result that holds in some of the examples
\begin{prop}
\label{Prop 4}Given the set of operators $\left\{ A_{n}\right\} _{n=1}^{M}$,
suppose that for some $n$ there exist a unitary operator $U$ such
$UA_{n}U^{\dagger}=-A_{n}$ and such that $\sum_{m\neq n}H_{m}$ is
left invariant by the adjoint action of $U\otimes U$, then 

$i)$ the ground state energy $\varepsilon_{gs,n}^{\alpha}$ of the
$H_{Tot,n}^{\alpha}$ defined in Proposition \ref{Prop-3} is an even
function of $\alpha$ i.e., $\varepsilon_{gs,n}^{\alpha}=\varepsilon_{gs,n}^{-\alpha}$; 

$ii)$ $\varepsilon_{gs,n}^{\alpha=0}$ is a local minimum for $\alpha$
varying in $\left[a_{n,1},a_{n,M}\right]$; 
\end{prop}
The Proof is given in Appendix \ref{Appendix: Proposition 4}. Result
$i)$ allows for each fixed $n$ to reduce the interval for the search
of $\min_{\alpha}\varepsilon_{gs,n}^{\alpha}$ to the positive interval
$\alpha\in\left[0,a_{nM}\right]$. Result $ii)$ allows to use Proposition
\ref{Prop-2} as a starting point for the minimization i.e., one could
first find $\varepsilon_{gs,n}^{\alpha=0}>0$ and use it as a first
estimate of the searched lower bound i.e., an upper bound of the global
minimum. \\We finally notice that in principle the mapping (\ref{eq: Mapping VTot HTot})
allows to enlarge the set symmetries that can be used to evaluate
the ground state of the specific Hamiltonian . Indeed, while the symmetries
of $V_{Tot}$ can obviously be translated into ones of the corresponding
Hamiltonian problem, there may be others $VH_{Tot}V=H_{Tot}$ represented
by unitary operators $V\neq U\otimes U$, which are not symmetries
of $V_{Tot}$, and that may of help in finding the ground state energy
and thus the desired lower bound.

\subsection{Strategy to find a state that (approximately) saturates the lower
bound.}

In order to complete our discussion, in the following we show how
it is possible, from the knowledge of the ground states to extract
further relevant information. Indeed, once the a state independent
lower bound $\tilde{l}_{B}^{-}$ has been found in terms of the ground
state energy of the operator under consideration, one is interested
on one hand in understanding how well $\tilde{l}_{B}^{-}$ approximate
the actual unknown optimal value $l_{B}$, and on the other hand in
identifying at least a state $\ket{\psi_{sat}}\in\mathcal{H}_{M}$
such that $V_{Tot}\left(\ket{\psi_{sat}}\right)\gtrapprox l_{B}$.
In this sub-section we describe how a state $\ket{\psi_{sat}}$ can
be in principle inferred and we discuss how its existence also provides
a way to check the goodness of the approximation $\tilde{l}_{B}^{-}$.
As shown above, in general the (non-trivial) lower bound will be found
in correspondence of the ground state $\ket{\varepsilon_{gs}}$ of
$H_{Tot}$, if $\varepsilon_{gs}\neq0,$ or in correspondence of the
ground state $\ket{\varepsilon_{gs,n}^{\alpha}}$ of some modified
version $H_{Tot,n}^{\alpha}$ for some fixed $\alpha$. In the following
discussion we drop for simplicity all indexes $\alpha,n$ and we refer
to a generic operator $H$ and relative ground state $\ket{\varepsilon}$
corresponding to $\varepsilon\neq0$. In general $\ket{\varepsilon}\neq\ket{\psi}\ket{\psi}$
i.e., the ground state is not in a product form and thus the bound
is not saturable. The strategy to find state $\ket{\psi_{sat}}\in\mathcal{H}_{M}$
is based on the Schmidt decomposition $\ket{\varepsilon}=\sum_{n}\lambda_{n}\ket{\lambda_{n}}\ket{\lambda_{n}^{'}}$,
where $\lambda_{n}\ge0$ are the Schmidt coefficients. If the ground
state is unique and the Schmidt coefficients are not degenerate, since
all of the above defined Hamiltonians are symmetric with respect to
a swap of the two identical Hilbert spaces onto which they are defined,
then $\ket{\lambda_{n}}=\ket{\lambda_{n}^{'}},\ \forall n$ i.e.,
the Schmidt decomposition is given in terms of product of \textit{identical}
states $\ket{\lambda_{n}}\ket{\lambda_{n}}$. The decomposition can
thus be used to find the desired $\ket{\psi_{sat}}$. Indeed if $\lambda_{Max}=\max_{n}\lambda_{n}$
a possible natural candidate for $\ket{\psi_{sat}}$ is the state
$\ket{\lambda_{Max}}$. For such state one has
\begin{eqnarray}
\bra{\lambda_{Max}}\bra{\lambda_{Max}}H\ket{\lambda_{Max}}\ket{\lambda_{Max}} & =\nonumber \\
\varepsilon\lambda_{Max}^{2}+\sum_{n=1}^{K}\varepsilon_{n}\left|\bra{\lambda_{Max}}\bra{\lambda_{Max}}\left.\varepsilon_{n}\right\rangle \right|^{2}\label{eq: Average H_Tot Shmidt}
\end{eqnarray}
where $\left\{ \varepsilon_{n},\ket{\varepsilon_{n}}\right\} _{n\ge1}$
are the eigenvalues and eigenstates of $H$ above the ground state,
and $K=M^{2}-1$. Unless $\ket{\varepsilon}=\ket{\lambda_{Max}}\ket{\lambda_{Max}}$
the sum for $n\ge1$ in (\ref{eq: Average H_Tot Shmidt}) is not negligible
such that the average $\bra{\lambda_{Max}}\bra{\lambda_{Max}}H\ket{\lambda_{Max}}\ket{\lambda_{Max}}>\varepsilon$
and it can in general be larger than $\varepsilon$. However, we can
upper bound the sum and to find some conditions on $\lambda_{Max}$
that guarantee that the average is sufficiently close to $\varepsilon$.
Given $\lambda_{Max}$, since $\varepsilon_{n}>0,\ \forall n\ $ then
the sum in (\ref{eq: Average H_Tot Shmidt})
\begin{eqnarray*}
\sum_{n=1}^{K}\varepsilon_{n}\left|\bra{\lambda_{Max}}\bra{\lambda_{Max}}\left.\varepsilon_{n}\right\rangle \right|^{2} & \le & \varepsilon_{K}\left(1-\lambda_{Max}^{2}\right)
\end{eqnarray*}
is upper bounded by the maximal eigenvalue $\varepsilon_{K}$. Therefore
the worst case scenario is given by 
\begin{eqnarray*}
\bra{\lambda_{Max}}\bra{\lambda_{Max}}H\ket{\lambda_{Max}}\ket{\lambda_{Max}} & = & \varepsilon\lambda_{Max}^{2}+\\
 & + & \varepsilon_{K}\left(1-\lambda_{Max}^{2}\right)
\end{eqnarray*}
Now in order for the state $\ket{\lambda_{Max}}\ket{\lambda_{Max}}$
to give a good approximation of $\varepsilon$ one has to impose that
$\varepsilon\lambda_{Max}^{2}\gg\varepsilon_{K}\left(1-\lambda_{Max}^{2}\right)$
or 
\begin{eqnarray}
\frac{\lambda_{Max}^{2}}{\left(1-\lambda_{Max}^{2}\right)} & \gg & \frac{\varepsilon_{K}}{\varepsilon}\label{eq: condition Shmidt Decomposition}
\end{eqnarray}
If one is able to determine $\lambda_{Max}^{2}$ and if the previous
condition is satisfied then 
\begin{eqnarray*}
\bra{\lambda_{Max}}\bra{\lambda_{Max}}H\ket{\lambda_{Max}}\ket{\lambda_{Max}} & \gtrapprox & \varepsilon\lambda_{Max}^{2}
\end{eqnarray*}
In the most favorable case $\lambda_{Max}\left(M\right)=O(1)$ and
$\lambda_{Max}\gg\lambda_{n},\ \forall\lambda_{n}\neq\lambda_{Max}$
i.e., $\lambda_{Max}$ is sufficiently larger than the other Schmidt
coefficients, such that one can identify $\ket{\psi_{sat}}=\ket{\lambda_{Max}}$.

The existence of $\ket{\psi_{sat}}$ allows for the desired assessment
of the goodness of the approximation provided by $\varepsilon$. Since
$V_{Tot}\left(\ket{\psi_{sat}}\right)=\bra{\lambda_{Max}}\bra{\lambda_{Max}}H\ket{\lambda_{Max}}\ket{\lambda_{Max}}\ge\varepsilon$
the actual unknown lower bound $l_{B}$ must lie in the interval $\left[\varepsilon,V_{Tot}\left(\ket{\psi_{sat}}\right)\right]$;
the smaller this interval the better the approximation. In the examples
described below we provide evidences that the above method can indeed
be successfully applied.

\section{Examples\label{sec: Examples}}

The examples that we present are different in many aspects, and we
use each of them to highlight different features of the scheme proposed
and how the latter can in principle be further modified. The first
two involve generators of the $su(2)$ algebra, and their relative
bounds have already been obtained in the literature. The other ones
are new. The third example involves $su(3)$ operators; this will
also allow us to compare the results obtainable with our approach
with those obtained with other methods \cite{ZycowskyNumricalRange}.
We finally use the fourth example to show how the mappings proposed
may be used even in the case unbounded operators. 

\subsection{Generators of $su(2)$\label{subsec: Generators-of-su(2)}}

In this first example we show a case in which the initial mapping
provided by $H_{Tot}$ is sufficient to obtain the desired lower bound;
and we also show how $H_{Tot}$ and $H_{Tot,n}^{\alpha}$ are just
starting points and different mappings are possible depending on the
specific problem at hand. We recover the bound for the sum of the
variances of the three generators $J_{X},J_{Y},J_{Z}$ of the $2j+1-$dimensional
irreducible representation of $su\left(2\right)$:
\begin{eqnarray}
V_{XYZ} & = & \Delta^{2}J_{X}+\Delta^{2}J_{Y}+\Delta^{2}J_{Z}\label{eq: V_XYZ}
\end{eqnarray}
The attainable lower bound of $l_{B}=j$ has already be found with
different methods \cite{URHoffmanEntanglementDetection,SIndepURMaccone}.
Here in principle the operator $H_{Tot}$ one needs to diagonalize
is 
\begin{eqnarray}
H_{Tot}= & \underset{{\scriptstyle \alpha=X,Y,Z}}{\sum} & \left(\frac{J_{\alpha}^{2}\otimes\mathbb{I}_{2j+1}+\mathbb{I}_{2j+1}\otimes J_{\alpha}^{2}}{2}-J_{\alpha}\otimes J_{\alpha}\right)\nonumber \\
\label{eq: HTot JX JZ}
\end{eqnarray}
It turns out that its ground state energy $\varepsilon_{gs}=j$ coincides
with $l_{B}$ and it is attained by the product ground states $\ket{j,j}_{z}\otimes\ket{j,j}_{z}$
and $\ket{j,-j}_{z}\otimes\ket{j,-j}_{z}$, such that the bound for
the variance is indeed attainable. In order to show how the method
we propose can be flexibly adapted to specific situations we obtain
the same result by means of a different mapping that makes use of
the following property of the $su(2)$ algebra. The Casimir operator
of the $su(2)$ algebra can be expressed as 
\begin{eqnarray*}
C & = & J_{X}^{2}+J_{Z}^{2}+J_{Z}^{2}=\\
 & = & j(j+1)\mathbb{I}_{2j+1}
\end{eqnarray*}
therefore, by using the previous relation, one can map the minimization
of the sum of variances
\begin{eqnarray*}
V_{XYZ} & = & j(j+1)-\left\langle J_{X}\right\rangle ^{2}-\left\langle J_{Y}\right\rangle ^{2}-\left\langle J_{Z}\right\rangle ^{2}
\end{eqnarray*}
into a new eigenvalue problem based on the operator 
\begin{eqnarray*}
H_{Tot}^{'} & = & j(j+1)\mathbb{I}_{2j+1}\otimes\mathbb{I}_{2j+1}-\sum_{\alpha=X,Y,Z}J_{\alpha}\otimes J_{\alpha}
\end{eqnarray*}
where again, for every state $\ket{\psi}\in\mathcal{H}_{2j+1}$ one
has $V_{XYZ}\left(\ket{\psi}\right)=\bra{\psi}\bra{\psi}H_{Tot}^{'}\ket{\psi}\ket{\psi}$.
Now the operator $H_{Heis}=-\sum_{\alpha=X,Y,Z}J_{\alpha}\otimes J_{\alpha}$
is well known since it represents a Heisenberg isotropic Hamiltonian
whose ferromagnetic ground states are for example $\ket{j,j}_{z}\otimes\ket{j,j}_{z}$
( $\ket{j,-j}_{z}\otimes\ket{j,-j}_{z}$) and they correspond to a
ground state energy $\varepsilon_{gs}^{Heis}=-j^{2}$ such that 
\begin{eqnarray}
\min V_{XYZ} & = & \bra{j,j}\bra{j,j}H_{Tot}^{'}\ket{j,j}\ket{j,j}=\nonumber \\
 & = & j\label{eq: V_XYZ_min_j/2}
\end{eqnarray}
The lower bound found is thus non-trivial and, since in this case
the ground states are product states, it is saturated by $\ket{\psi_{sat}}=\ket{j,j},\ket{-j,-j}$.
It is then easy to check that the states $\ket{j,j}_{z}\otimes\ket{j,j}_{z}$
and $\ket{j,-j}_{z}\otimes\ket{j,-j}_{z}$ are also ground states
of $H_{Tot}$ and that they correspond to the ground state energy
$\varepsilon_{gs}=j$. \\This first result shows on one hand that
the mapping (\ref{eq: Mapping VTot HTot}) introduced in the previous
Section can directly provide the desired lower bound in terms of $\varepsilon_{gs}$.
On the other hand, it shows that by using the information about the
relations between the operators involved in $V_{XYZ}$, in this case
the algebraic relation provided by the Casimir, one can find another
mapping that allows to derive the desired lower bound as the solution
of a known eigenvalue problem.

\subsection{Spin operators and planar squeezing\label{subsec:Spin-operators-and}}

We now focus on an example that allows us to illustrate many of the
results derived in the previous section. We first derive the lower
bound by selecting the relevant Hamiltonian on the basis of symmetry
arguments. We then discuss how one can find the state $\ket{\psi_{sat}}$
able to fairly well approximate the bound and we show that the $\ket{\psi_{sat}}$
we identify is in principle obtainable in the laboratory via two-axis
spin squeezing \cite{KitagawaSpinSqueezing,NoriSpinSqueezReview}.
\\We focus on a pair of generators of $su(2)$. In order to fix the
ideas and without loss of generality we choose to work with 
\begin{eqnarray}
V_{XZ} & = & \Delta^{2}J_{X}+\Delta^{2}J_{Z}\label{eq: VXZ}
\end{eqnarray}
The minimization of $V_{XZ}$ has been introduced in \cite{HePQS},
where it was shown that the simultaneous reduction of the noise $V_{XZ}$
of two orthogonal spin projections in the plane $XZ$ (e.g. $J_{X},J_{Z}$)
can be relevant for the optimization one-shot phase measurements,
since it allows for phase uncertainties $\Delta\phi\sim j^{-2/3}$
i.e., a precision beyond the standard quantum limit, that importantly
do not depend on the actual value of the phase $\phi$ \cite{HePQSEntangInterfero,PQSExpermMitchell,PQSNDMeasMitchell}.
In \cite{HePQS} the behaviour of $V_{XZ}$ in the asymptotic limit
$j\rightarrow\infty$ was obtained by means of analytical arguments
and the overall behaviour of $V_{XZ}^{min}\left(j\right)$ via numerical
fitting such that 
\begin{eqnarray}
V_{XZ}^{min_{1}}\left(j\right) & \simeq & 0.595275\ j^{2/3}-0.1663\ j^{1/3}+0.0267\nonumber \\
\label{eq: VXZ He scaling}
\end{eqnarray}
On the other hand, in \cite{WerAngularMomentum} the asymptotic behaviour
was obtained numerically by means of a \textit{seesaw} algorithm as
\begin{eqnarray}
V_{XZ}^{min_{2}}\left(j\right) & \approx & 0.569524\ j^{2/3}\label{eq: VXZ Werner scaling}
\end{eqnarray}
We start our analysis by showing that the Hamiltonian 
\begin{eqnarray*}
H_{Tot}= & \sum_{\alpha=X,Z} & \left(\frac{J_{\alpha}^{2}\otimes\mathbb{I}_{2j+1}+\mathbb{I}_{2j+1}\otimes J_{\alpha}^{2}}{2}-J_{\alpha}\otimes J_{\alpha}\right)
\end{eqnarray*}
has ground state energy is zero. Indeed, $\forall j$ one can write
\begin{eqnarray*}
\ket{\varepsilon_{gs}} & = & \frac{1}{\sqrt{2j+1}}\sum_{m_{z}=-j}^{j}\ket{j,m_{z}}\ket{j,m_{z}}=\\
 & = & \frac{1}{\sqrt{2j+1}}\sum_{m_{x}=-j}^{j}\ket{j,m_{x}}\ket{j,m_{x}}
\end{eqnarray*}
and check that $\varepsilon_{gs}=0$. One can subsequently use result
$ii)$ in Proposition \ref{Prop 1} and evaluate $\varepsilon_{1}\left(1-\frac{1}{2j+1}\right)$.
However in this case one can easily check that $\varepsilon_{1}=0.5$
for all $j$ and therefore $H_{Tot}$ provides a non-zero lower bound
which scales poorly with $j$. We are thus led to use the strategy
based on the Hamiltonians $H_{Tot,n}^{\alpha}$ described in Proposition
\ref{Prop-3}. This is however a case in which we can apply Proposition
\ref{Prop 4}. Indeed, one has that $U=\exp\left(-i\pi J_{Z}\right)$
is such that $UJ_{X}U^{\dagger}=-J_{X}$ and the adjoint action of
$U\otimes U$ obviously leaves the whole Hamiltonian $H_{Tot}$ invariant.
Therefore one can start by searching for the lower bound among the
states belonging to the set $S_{X}^{0}=\left\{ \ket{\psi}\in\mathcal{H}_{2j+1}|\bra{\psi}J_{X}\ket{\psi}=0\right\} $
and use the Hamiltonian
\begin{eqnarray}
H_{Tot,X} & = & \sum_{\alpha=X,Z}\left(\frac{J_{\alpha}^{2}\otimes\mathbb{I}_{2j+1}+\mathbb{I}_{2j+1}\otimes J_{\alpha}^{2}}{2}\right)-J_{Z}\otimes J_{Z}\nonumber \\
\label{eq: HTotX}
\end{eqnarray}
The relative lower bound $\varepsilon_{gs,X}^{0}$ provides a local
minimum. Then one should extend the search by using the Hamiltonian
$H_{Tot,X}^{\alpha}$ with $\alpha\in\left[0,j\right]$. Of course
this strategy is of use when $j$ is sufficiently small, whereas $j$
becomes large the task would be quite demanding. However, in this
case the search in $S_{X}^{0}$ is sufficient to obtain the overall
lower bound since the Hamiltonian $H_{Tot}$ enjoys the same type
of continuous symmetry of $V_{Tot}$. Indeed $V_{Tot}\left[\ket{\psi}\right]=V_{Tot}\left[\exp\left(i\theta J_{Y}\right)\ket{\psi}\right]$
for all $\ket{\psi}$ and $\theta\in\mathbb{R}$ and in the same way
given $U_{YY}=\exp\left(-i\theta J_{Y}\right)\otimes\exp\left(-i\theta J_{Y}\right)$
\begin{eqnarray*}
U_{YY}H_{Tot}U_{YY}^{\dagger} & = & H_{Tot}
\end{eqnarray*}
and this allows to limit the minimization over $S_{X}^{0}$ \cite{HePQS,WerAngularMomentum}
(see also Appendix \ref{sec: Appendix: symmetry for spin Ham}). Furthermore
since the role of $Z$ and $X$ can be exchanged we can focus on $H_{Tot,X}$
only. We notice that, when expressed in the $J_{Z}$ eigenbasis, $H_{Tot,X}$
is banded and sparse and thus efficient algorithms can be used for
its diagonalization. The ground state energy $\varepsilon_{gs,X}(j)$
can then be numerically evaluated for different values of $j$, it
is always non-zero and the results are plotted in Fig. (\ref{fig: VarXZ GROUND STATE SCALING}
- left panel) and compared with the two bounds (\ref{eq: VXZ He scaling})
and (\ref{eq: VXZ Werner scaling}). The result show that $\forall j$
$\varepsilon_{gs,X}(j)\le V_{XZ}^{min_{1}}\left(j\right)\le V_{XZ}^{min_{2}}\left(j\right)$
and the ground state energy of $H_{Tot,X}$ provide a fairly good
and meaningful lower bound. 

\begin{figure}[h]
\subfigure{\includegraphics[scale=0.8]{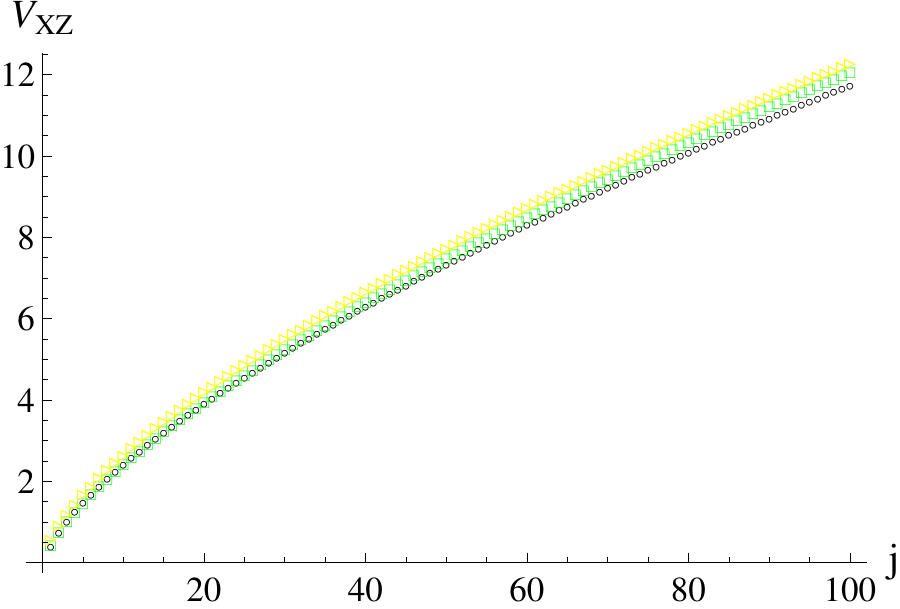} }
\subfigure{ \includegraphics[scale=0.8]{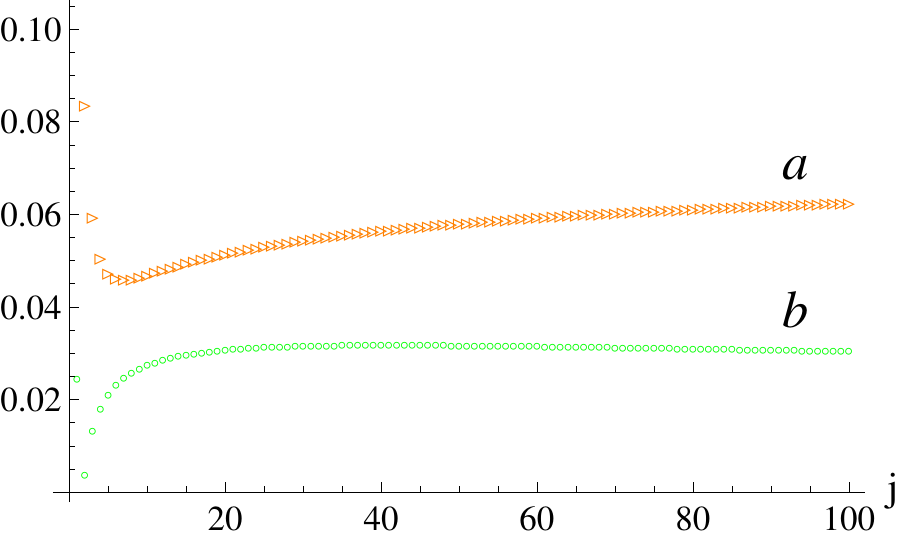}}\caption{\textbf{Left: }Scaling of the sum of variances $V_{XZ}$ with $j=(1,100)$:
(black circles) lower bound of $V_{XZ}$ provided by the ground state
energy $\varepsilon_{gs,X}\left(j\right)$ of the Hamiltonian (\ref{eq: HTotX});
(green squares ) $V_{XZ}^{min_{1}}\left(j\right)$ as in (\ref{eq: VXZ He scaling});
( yellow triangles) $V_{XZ}^{min_{2}}\left(j\right)$ as in (\ref{eq: VXZ Werner scaling}).
\textbf{Right: }Relative errors obtained with the use of $\ket{\theta_{m}}=\exp\left(-i\theta_{m}H_{TAS}\right)\ket{j,j}$
(see text) as a function of $j=1,..,100$. Curve \textbf{a} (triangles)
$|V_{Tot}(\ket{\theta_{m}})-\varepsilon_{gs,X}|/\varepsilon_{gs,X};$
curve \textbf{b} (circles) $|V_{Tot}(\ket{\theta_{m}})-V_{XZ}^{min_{1}}\left(j\right)|/V_{XZ}^{min_{1}}\left(j\right)$
\label{fig: VarXZ GROUND STATE SCALING}}
\end{figure}

The algorithm implemented requires the diagonalization process that
eventually determines the value of the bound. However, the structure
of the state $\ket{\psi_{sat}}$ able to approximately saturate the
bound is not directly apparent from the algorithm unless the ground
state is a product state $\ket{\varepsilon_{gs,X}}=\ket{\psi}\ket{\psi}$.
In this case, the numerical computations suggest that the ground state
is not a in a product form although it provides values which are pretty
close to those evaluated in (\ref{eq: VXZ He scaling}). The results
obtained can be refined in the following way. For generic $j$ one
has that the numerical found ground state energy is doubly degenerate.
By fixing $j$ one can explore the ground state manifold in search
for a ground state whose Schmidt decomposition can be written as $\ket{\varepsilon_{gs,X}}=\sum_{n}\lambda_{n}\ket{\lambda_{n}}\ket{\lambda_{n}}$
and such that the maximum Schmidt coefficient is sufficiently large.
For fixed $j$ we can identify two states $\ket{\lambda_{Max}^{+}},\ket{\lambda_{Max}^{-}}$
corresponding to two different states $\ket{\varepsilon_{gs,X}^{+}},\ket{\varepsilon_{gs,X}^{-}}$
both belonging to the ground state manifold and for which the largest
Schmidt coefficients coincide. For example with $j=9/2$ one finds
sufficiently large values $\lambda_{Max}^{+}=\lambda_{Max}^{-}=0.99619$
. The overlap of the product states with the respective ground states
is equal and large i.e., $\left\langle \varepsilon_{gs,+}^{XZ}\right.\ket{\lambda_{Max}^{+}}\ket{\lambda_{Max}^{+}}=\left\langle \varepsilon_{gs,-}^{XZ}\right.\ket{\lambda_{Max}^{-}}\ket{\lambda_{Max}^{-}}=0.996191$.
Similar results have be obtained for generic values of $j\le100$,
thus one one hand both states $\ket{\lambda_{Max}^{+}},\ket{\lambda_{Max}^{-}}$
constitute good candidates for $\ket{\psi_{sat}}$ and for the (approximate)
saturation of the found lower bound, and on the other hand the result
is an indirect confirmation that the lower bound provided by $\varepsilon_{gs,X}$
is close to the actual one $l_{B}$. 

In order to estimate the error in determining the lower bound via
$\varepsilon_{gs,X}$ i.e., $V_{Tot}\left(\ket{\psi_{sat}}\right)-\varepsilon_{gs,X}$,
we now proceed with a further refined approach to determine $\ket{\psi_{sat}}$.
Indeed, while the states $\ket{\lambda_{Max}^{\pm}}$, which are good
candidates for $\ket{\psi_{sat}}$, are obtained numerically it would
be desirable to find analogous states that at least in principle can
be produced in the laboratory, and that have the same property of
$\ket{\lambda_{Max}^{\pm}}$ i.e..to approximately saturate the lower
bound. In Appendix \ref{Appendix: Planar Spin Squeezing} we show
how starting from the knowledge of the shape of $\ket{\lambda_{Max}^{\pm}}$
and by means of further physical insights one can indeed identify
the following candidate
\begin{eqnarray*}
\ket{\theta} & = & \exp\left(-i\theta H_{TAS}\right)\ket{j,j}
\end{eqnarray*}
where: $\ket{j,j}$ is the eigenstate of $J_{Z}$ corresponding to
the eigenvalue $j$; and 
\begin{eqnarray*}
H_{TAS} & = & -i\left(J_{+}^{2}-J_{-}^{2}\right)
\end{eqnarray*}
is the two-axis squeezing operator \cite{KitagawaSpinSqueezing,NoriSpinSqueezReview};
the latter having the property of squeezing the state along the $X$
axis and simultaneously anti-squeeze it along the $Y$ axis. As shown
in Appendix \ref{Appendix: Planar Spin Squeezing}, by means of the
mapping provided by the Holstein Primakoff approximation, it is possible
to infer the optimal value of the squeezing parameter $\theta_{m}=-\frac{\log2+\log j}{24\ j}$
such that $\ket{\psi_{sat}}=\ket{\theta_{m}}$ provides a good approximation
of the lower bound for each $j$. In Figure \ref{fig: VarXZ GROUND STATE SCALING}
(right panel, curve b) we plot $|V_{Tot}(\ket{\theta_{m}})-V_{XZ}^{min_{1}}\left(j\right)|/V_{XZ}^{min_{1}}\left(j\right)$
i.e., the relative error in the evaluation of $V_{Tot}$ with respect
to the best bound given by $V_{XZ}^{min_{1}}\left(j\right)$. For
$j\le100$ the error is firmly below $3\%$, thus showing that the
approximation provided by $\ket{\theta_{m}}$ is indeed quite good.
\\With the aid of $\ket{\theta_{m}}$ we can then provide an estimate
of the errors in the determination of the lower bound by means of
$\varepsilon_{gs,X}$. In Figure \ref{fig: VarXZ GROUND STATE SCALING}
(right panel, curve a) we plot $|V_{Tot}(\ket{\theta_{m}})-\varepsilon_{gs,X}|/\varepsilon_{gs,X}$;
the latter shows that the relative error is for $j\le100$ of the
order of $6\%$, a result that confirms the goodness of the approximation
provided by $\varepsilon_{gs,X}$. Similar results can be obtained
directly using $\ket{\lambda_{Max}^{+}},\ket{\lambda_{Max}^{-}}$
instead of $\ket{\theta_{m}}$. \\We finally notice that the state
$\ket{\theta_{m}}$ is in principle obtainable in the laboratory via
two-axis squeezing and thus is a good candidate for the estimation
procedure based on Planar Squeezed states. While the realization of
the latter has been proposed in \cite{HePQS} as the ground state
of a two-mode Bose-Einstein condensate and in \cite{PQSNDMeasMitchell}
as the result of a non-demolition quantum measurement protocol, here
\textit{we provide evidence that the same result can be obtained via
two-axis spin-squeezing}.

\subsection{su(3) operators\label{subsec:su(3)-operators}}

We now derive the lower bound for the sum of the variances of $4$
operators belonging to the $su(3)$ algebra. This will allow us to
show the results of Proposition \ref{Prop-3} in action. Consider
the following operators 
\begin{eqnarray*}
A_{1}=\left(\begin{array}{ccc}
0 & 1 & 0\\
1 & 0 & i\\
0 & -i & 0
\end{array}\right), &  & A_{2}=\left(\begin{array}{ccc}
1 & 0 & 0\\
0 & 0 & 0\\
0 & 0 & -1
\end{array}\right)\\
A_{3}=\left(\begin{array}{ccc}
1 & 1 & 0\\
1 & 0 & -1\\
0 & -1 & -1
\end{array}\right), &  & A_{4}=\left(\begin{array}{ccc}
1 & 0 & i\\
0 & 0 & 0\\
-i & 0 & -1
\end{array}\right)
\end{eqnarray*}
The bounds for the sum of pair of variances $V_{12}=\Delta^{2}A_{1}+\Delta^{2}A_{2}\ge15/32$
and $V_{34}=\Delta^{2}A_{3}+\Delta^{2}A_{4}\ge0.765727$ were found
in \cite{ZycowskyNumricalRange} on the basis of the \textit{(uncertainty)
numerical range} approach. If we compare these results with the approximations
$\tilde{l}_{B}^{-}$ obtained within our framework we find that: for
$V_{12}$, $\tilde{l}_{B}^{-}=0.4384$ which is approximately $6.5\%$
lower that the value found in \cite{ZycowskyNumricalRange}; while
for $V_{34}$, $\tilde{l}_{B}^{-}=0.7281$ which is approximately
$5\%$ lower that the value found in \cite{ZycowskyNumricalRange}.
As for the lower bound of the sum of the four variances $V_{Tot}=\Delta^{2}A_{1}+\Delta^{2}A_{2}+\Delta^{2}A_{3}+\Delta^{2}A_{4}$
the ground state energy of the corresponding $H_{Tot}$ is different
from zero and it provides a first approximation of the searched lower
bound i.e., $\varepsilon_{gs}=0.804103$. The problem does not appear
to have evident symmetries and in order to check the consistency of
$\varepsilon_{gs}$ and to refine the approximation we then use the
method outlined in Proposition \ref{Prop-3}. In Figure (\ref{fig: Var1234 su(3)})
we plot the values of the ground states $\varepsilon_{gs,n}^{\alpha}$
of the Hamiltonians $H_{Tot,n}^{\alpha},\ n=1,2,3,4$ as a function
of $\alpha\in\left[a_{n1,}a_{n3,}\right]$ i.e., $\alpha$ varying
in the interval defined by the lowest/highest eigenvalue of each $A_{n}$.
The best lower bound $\tilde{l}_{B}^{-}=\max_{n}\min_{\alpha}\varepsilon_{gs,n}^{\alpha}$is
obtained with the Hamiltonian $H_{Tot,1}^{\alpha}$ in correspondence
of the value $\alpha=0.963$. The corresponding lower bound $\tilde{l}_{B}^{-}=\varepsilon_{gs,1}^{\alpha=0.963}=1.39932$
is higher than $\varepsilon_{gs}=0.804103$, therefore showing that
the method outlined in Proposition \ref{Prop-3} allows for a significative
refinement of the result. If we now find the Schmidt decomposition
of $\ket{\varepsilon_{gs,1}^{\alpha=0.963}}$, we have that the largest
Schmidt coefficient is $\lambda_{Max}=0.941487$ and for the corresponding
$\ket{\lambda_{Max}}$ the value of $V_{Tot}\left(\ket{\lambda_{Max}}\right)=1.5901$.
Therefore the actual bound $l_{B}$ will lie in the interval $\left(\varepsilon_{gs,1}^{\alpha=0.963},V_{Tot}\left(\ket{\lambda_{Max}}\right)\right]=\left(1.39932,1.5901\right]$.
Since the Hilbert space has dimension $3$ we have performed a standard
minimization procedure directly on $V_{Tot}$ and we have obtained
$l_{B}\approx1.56274$ such that: $\varepsilon_{gs}$ is about half
the value $l_{B}$; $\varepsilon_{gs,1}^{\alpha=0.963}$ results to
be smaller for about $10\%$; while $V_{Tot}\left(\ket{\lambda_{Max}}\right)$
is just $1.6\%$ higher.

\begin{figure}
\subfigure{\includegraphics[scale=0.48]{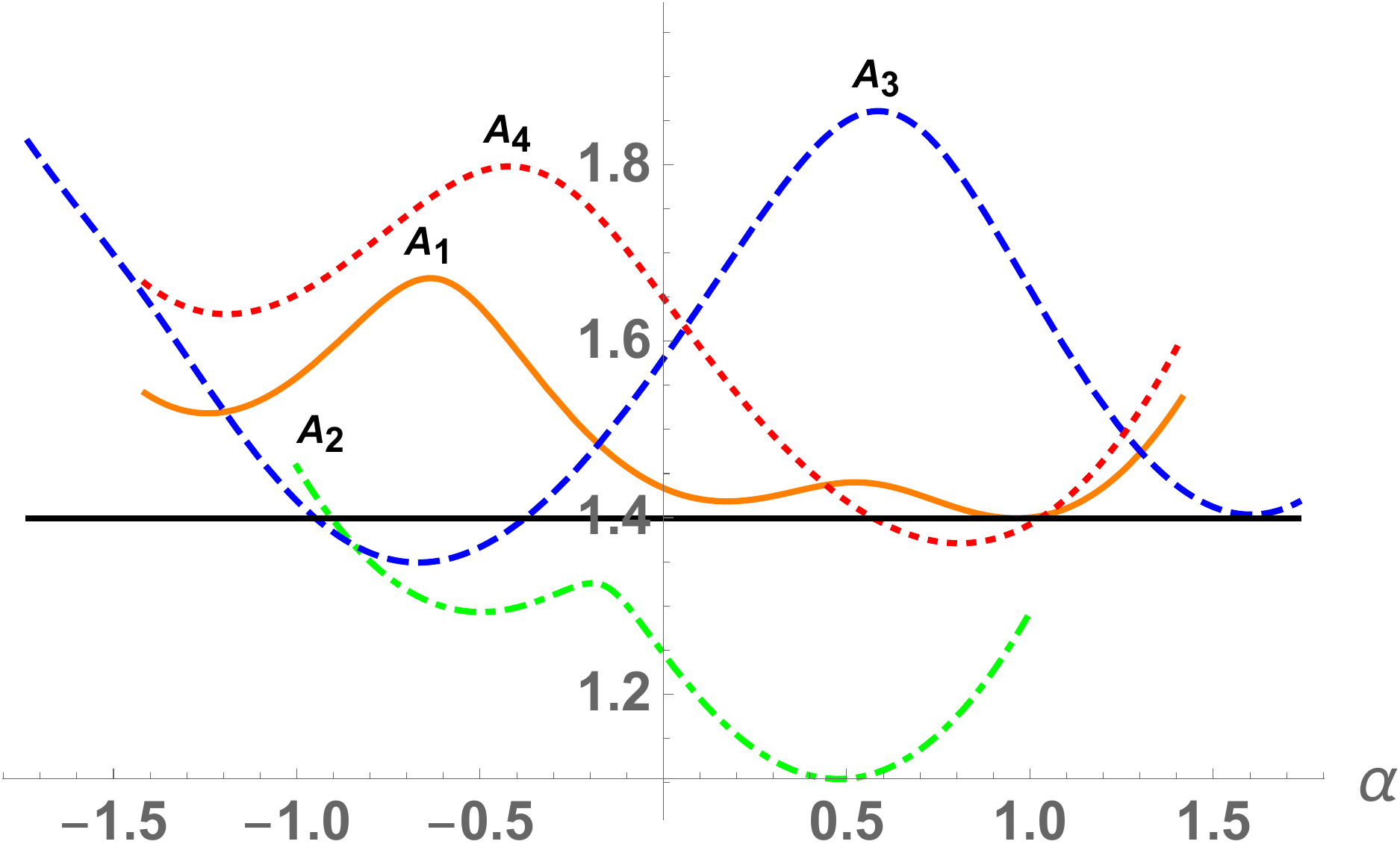}}\caption{Plot of $\varepsilon_{gs,n}^{\alpha}$ as a function of $\alpha\in\left[a_{n1,}a_{n4,}\right]$
for the operators $A_{1}$ (orange continuous), $A_{2}$ (dash dotted)
, $A_{3}$ (dashed) , $A_{4}$ (dotted) . The best lower bound $\varepsilon_{gs,1}^{\alpha=0.963}$
is attained for $H_{Tot,1}^{\alpha=0.963}$ (black continuous horizontal)
\label{fig: Var1234 su(3)}}
\end{figure}

\subsection{Harmonic oscillator operators $\hat{n},\hat{x}$\label{subsec:Harmonic-oscillator-operators}}

While the definition of $H=\sum_{n}H_{n}$ was given for bounded operators,
one can use the same definition for unbounded one and use the same
mapping (\ref{eq: Mapping VTot HTot}), which of course remains valid,
for finding the relative lower bounds. In the following we show how
the procedure and the results of Section \ref{sec: General-Results}
can be applied by focusing a specific example. We consider the operators
$\hat{n}$ (number operator) and $\hat{x}$ (position operator) for
a single bosonic mode and we seek for the lower bound of 
\begin{eqnarray}
V_{xn} & = & \Delta^{2}\hat{n}+\Delta^{2}\hat{x}\label{eq: Var_xn}
\end{eqnarray}
The latter is very much analogous to the bosonic counterpart of $V_{XZ}$
with $j=1$, see equation (\ref{eq: VXZBoson}) in Appendix \ref{Appendix: Planar Spin Squeezing}.
The analogy with the spin case is strengthened by the three variances
sum 
\begin{eqnarray*}
V_{xpn} & = & \Delta^{2}\hat{n}+\Delta^{2}\hat{x}+\Delta^{2}\hat{p}\ge1
\end{eqnarray*}
whose lower bound is again attained by the analog of $\ket{j,j}$
i.e., the vacuum $\ket{0}$ for which $V_{xpn}=1$ and $V_{xn}=1/2$.
If one is to reduce $V_{xn}$ one needs to simultaneously reduce $\Delta^{2}\hat{x}<1/2$
and therefore enhance $\Delta^{2}\hat{p}>1/2$. \\The starting Hamiltonian
here is 
\begin{eqnarray*}
H_{Tot} & = & \frac{1}{2}\left(\hat{n}^{2}\otimes\mathbb{I}+\mathbb{I}\otimes\hat{n}^{2}\right)-\hat{n}\otimes\hat{n}+\\
 & + & \frac{1}{2}\left(\hat{x}^{2}\otimes\mathbb{I}+\mathbb{I}\otimes\hat{x}^{2}\right)-\hat{x}\otimes\hat{x}
\end{eqnarray*}
and its approximate ground state energy can be found by expressing
$\hat{x}=\left(a+a^{\dagger}\right)/\sqrt{2}$ and by truncating the
single mode Fock space i.e., by expressing $H_{Tot}$ in the subspace
$\mathcal{H}_{n_{Max}}\otimes\mathcal{H}_{n_{Max}}$ with $\mathcal{H}_{n_{Max}}=span\left\{ \ket{0},\ket{1},..,\ket{n_{Max}}\right\} $
where $\ket{n}$ is an $n$ bosons state. By letting the maximum number
of bosons $n_{Max}$ grow we numerically check that $\varepsilon_{gs}\rightarrow0$,
therefore $H_{Tot}$ itself does not provide a meaningful lower bound.
However here we can again resort to the result of Proposition \ref{Prop 4}
and thus identify the needed modified Hamiltonian. Indeed, the relevant
unitary operator here is $U_{\theta}=\exp\left(-i\theta\hat{n}\right)$;
one has that $U_{\pi}\hat{x}U_{\pi}^{\dagger}=-\hat{x}$, and the
adjoint action of $U_{\pi}\otimes U_{\pi}$ leaves the Hamiltonian
$H_{Tot}$ invariant. Therefore, in search for the lower bound we
can start restricting ourselves to the states belonging to $S_{\hat{x}}^{0}=\left\{ \ket{\psi}\in\mathcal{H}_{bos}|\left\langle \hat{x}\right\rangle =0\right\} $
and consider the Hamiltonian 
\begin{eqnarray*}
H_{Tot,\hat{x}} & = & \frac{1}{2}\left(\hat{n}^{2}\otimes\mathbb{I}+\mathbb{I}\otimes\hat{n}^{2}\right)-\hat{n}\otimes\hat{n}+\\
 & + & \frac{1}{2}\left(\hat{x}^{2}\otimes\mathbb{I}+\mathbb{I}\otimes\hat{x}^{2}\right)
\end{eqnarray*}
and its ground state energy $\varepsilon_{gs,\hat{x}}^{0}$ which
is a local minimum. For sufficiently high values of $n_{Max}$ one
has that $\varepsilon_{gs,\hat{x}}^{0}$ converges to the value $\varepsilon_{gs,\hat{x}}^{0}\approx0.412721<1/2$.
The ground state in this case $\ket{\varepsilon_{gs,\hat{x}}^{0}}\neq\ket{\psi}\ket{\psi}$
is not in a product form, however we can again use the argument outlined
in Section \ref{sec: General-Results} and find the Schmidt decomposition
$\ket{\varepsilon_{gs,\hat{x}}^{0}}=\sum_{n}\lambda_{n}\ket{\lambda_{n}}\ket{\lambda_{n}}$.
For $n_{Max}=30$ we have that the maximum Schmidt coefficient $\lambda_{Max}\approx0.99931$
such that one is led to consider the corresponding state $\ket{\lambda_{Max}}\ket{\lambda_{Max}}$
as a fairly good approximation of the ground state. Indeed $\left|\left\langle \varepsilon_{gs,\hat{x}}^{0}\right.\ket{\lambda_{Max}}\ket{\lambda_{Max}}\right|\approx0.99931$
and therefore $\ket{\psi_{sat}}=\ket{\lambda_{Max}}$ in this case
is a good candidate for the minimization of (\ref{eq: Var_xn}). This
is confirmed by the value $V_{xn}\left(\ket{\lambda_{Max}}\right)\approx0.415139$
such that the relative error of the approximation $\left|V_{xn}\left(\ket{\lambda_{Max}}\right)-\varepsilon_{gs,\hat{x}}^{0}\right|/\varepsilon_{gs,\hat{x}}^{0}\approx0.5\%$
is excellent. While the previous results have been obtained numerically,
the following arguments allow one to identify a state realizable in
the laboratory that closely approximate $\ket{\lambda_{Max}}$. Just
as in the spin case the profile of $\ket{\lambda_{Max}}=\sum_{n=0}^{n_{Max}}\eta_{n}\ket{n}$
is such that only the states with even number of bosons are populated,
the distribution of probability is peaked for $n=0$ and it rapidly
decreases with $n$. As in the $J_{X},J_{Z}$ case this again hints
to the preferred tentative choice of the single mode squeezed state
\begin{eqnarray*}
\ket{\xi} & = & \frac{1}{\sqrt{\cosh\left|\xi\right|}}\sum_{n=0}^{\infty}\left(-\tanh\left|\xi\right|\right)^{n}\frac{\sqrt{\left(2n\right)!}}{2^{n}n!}\ket{2n}
\end{eqnarray*}
 as candidate for the minimization of $V_{xn}$. Indeed, in terms
of $\ket{\xi}$ (\ref{eq: Var_xn}) reads 
\begin{eqnarray}
V_{xn} & = & 2\sinh^{2}(\left|\xi\right|)\cosh^{2}(\left|\xi\right|)+\frac{\exp\left(-2\left|\xi\right|\right)}{2},\label{eq: Vxn for single squeezed vacuum}
\end{eqnarray}
its minimum is obtained for $\xi=\xi_{m}=0.1665679$ and it is equal
to $V_{xn}(\ket{\xi_{m}})=0.41591$ which is a fairly good approximation
of $\varepsilon_{gs,\hat{x}}$ and $V_{xn}\left(\ket{\lambda_{Max}}\right)$.
Indeed, if one evaluates the fidelity between $\ket{\xi_{m}}$ and
the numerically obtained $\ket{\lambda_{Max}}$ one has $\left\langle \xi_{m}\left|\lambda_{Max}\right.\right\rangle =0.999927$;
furthermore $\left|\left\langle \varepsilon_{gs,\hat{x}}^{0}\right.\ket{\xi_{m}}\ket{\xi_{m}}\right|=0.999168$
such that $\ket{\xi_{m}}\ket{\xi_{m}}$ also provides a good approximation
of the ground state. 

Now in principle in order to find whether $\varepsilon_{gs,\hat{x}}^{0}$
is a proper and faithful lower bound one should extend the search
to the other sets $S_{\hat{x}}^{\alpha}$, $\alpha\in\left[0,\infty\right]$,
which is of course an impossible task. We thus opt for a different
strategy. In the first place, the result can be further supported
analytically by showing that $\ket{\xi_{m}}$ minimizes $V_{xn}$
over the restricted set of Gaussian states; this is shown in Appendix
\ref{sec: Appendix: bosonic-case xn}. Since the minimum corresponds
to $\ket{\xi_{m}}$ with $\left\langle n\right\rangle $ very small,
we further support our result by using standard numerical minimization
routines and search for the minimum of $V_{xn}$ in a sub space $H_{n_{Max}}=span\left\{ \ket{0},\ket{1},..,\ket{n_{Max}}\right\} $
with $n_{Max}$ sufficiently large; the numerical results rapidly
converge to the lower bound found above. 

We have thus shown how the method proposed can in principle work even
with sums of variances involving unbounded operators. With the analysis
of the Schmidt decomposition of the ground state $\ket{\varepsilon_{gs,\hat{x}}^{0}}$,
and the subsequent reasonings and calculations, we have shown that
is possible to identify a state that approximately saturates the bound
provided by $\varepsilon_{gs,\hat{x}}^{0}$. Therefore even in this
case the latter can be considered a good approximation of the actual
bound $l_{B}$.

\section{Conclusions}

In this work we have addressed the problem of finding the state independent
lower bound $l_{B}$ of the sum of variances $V_{Tot}\left(\ket{\psi}\right)=\sum_{1}^{N}\Delta_{\ket{\psi}}^{2}A_{n}$
for an arbitrary set $\left\{ A_{n}\right\} _{n=1,..,N}$ of Hermitian
operators acting on an Hilbert space $\mathcal{H}_{M}$ with dimension
$M$. The value $l_{B}$ is the highest positive constant such that
$\forall\ket{\psi}\in\mathcal{H_{M}},\ V_{Tot}\left(\ket{\psi}\right)\ge l_{B}$.
In general the problem can be solved by finding a sufficiently good
approximation $\tilde{l}_{B}^{-}\le l_{B}$. To this aim we have introduced
a method based on a mapping of the minimization problem into the task
of finding the ground state energy $\varepsilon_{gs}$ of specific
Hamiltonians acting on an extended space $\mathcal{H}_{M}\otimes\mathcal{H}_{M}$.
In such way we have shown that $\varepsilon_{gs}=\tilde{l}_{B}^{-}$
i.e., $\varepsilon_{gs}$ provides the required approximation.\\
In our work we have first provided the main general results that characterize
the method proposed and then, by means of different examples, we have
described its implementation. While we have shown an instance where
$\varepsilon_{gs}=l_{B}$, in general the ground state $\ket{\varepsilon_{gs}}\in\mathcal{H}_{M}\otimes\mathcal{H}_{M}$
corresponding to $\varepsilon_{gs}$ is not in a product form, such
that the corresponding $\varepsilon_{gs}=\tilde{l}_{B}^{-}<l_{B}$
will only be an approximation of the actual $l_{B}$, and the bound
provided by $\varepsilon_{gs}$ will not be attainable, even though
it will still be a valid state independent lower bound. In such cases
we have also proposed and tested a method to identify, from the knowledge
of the ground state $\ket{\varepsilon_{gs}}\in\mathcal{H}_{M}\otimes\mathcal{H}_{M}$,
a state $\ket{\psi_{sat}}\in\mathcal{H}_{M}$ that allows, at least
approximately, to saturate the bound i.e., $V_{Tot}\left(\ket{\psi_{sat}}\right)\gtrapprox l_{B}$
. This procedure provides an efficient way to assess the quality of
the approximations given by $\varepsilon_{gs}$ and $V_{Tot}\left(\ket{\psi_{sat}}\right)$:
the true lower bound $l_{B}$ must lie in the interval $\left(\varepsilon_{gs},V_{Tot}\left(\ket{\psi_{sat}}\right)\right]$.
The examples developed show that the latter can be very small, such
that even when $\varepsilon_{gs}\neq l_{B}$ the approximations are
quite good. While the main general results have been derived for bounded
(non-degenerate) operators, we have also shown by means of an example,
that the method can be applied to sum of variances involving unbounded
operators.\\ The results presented constitute a first attempt to
lay down a general and reliable framework, alternative to the existing
ones, for deriving meaningful state independent lower bounds for the
sum of variances $V_{Tot}$. As such we have discussed the virtues
and limits of the proposed framework. Since the latter is based on
ground states evaluation, it does not suffer from the caveats of general
minimization schemes that can be numerically demanding and can get
trapped in local minima. On the other hand it requires the diagonalization
of operators of dimension $M^{2}\times M^{2}$, that for $M$ very
large can be numerically complex. As we have shown the complexity
of the solution may however be drastically reduced when the problem
presents some symmetries and/or the operator involved are simple (e.g.
sparse). \\ While the examples discussed show that the method can
indeed be effective, several questions remain open for future research.
As we have shown in the paper, since the mapping is not unique, other
possibly more effective mappings may be found. The extension of the
method to cases involving unbounded operators and the assessment of
its limits require a thorough analysis. On another level it would
be intriguing to explore the connections, if any, between the framework
proposed and the already existing ones e.g. those based on the joint
numerical range.\\ Finally, while in this paper we have not assessed
the problem, our method can be used for entanglement detection \cite{URHoffmanEntanglementDetection,URGuneEntanglementDetection}
and it would be desirable to apply it to relevant problems in that
area of research.
\begin{acknowledgments}
We gratefully acknowledge funding from the University of Pavia ``Blue
sky'' project - grant n. BSR1718573. P. Giorda would like to thank
Professor R. Demkowicz-Dobrza\'{n}ski, Professor M.G.A. Paris. 
\end{acknowledgments}

\appendix

\appendix
\numberwithin{equation}{section}

\section{Properties of $H_{Tot}$\label{sec: Appendix Properties of H_Tot}}

In the following we prove point $ii)$ of Proposition \ref{Prop 1}
by construction. To this aim we start by supposing that each $A_{n}$
has a non-degenerate eigenspectrum. This hypothesis is in principle
not necessary but we use it to simplify the notations. We thus notice
that given a state $\ket{\phi}\in\mathcal{H}_{M}\otimes\mathcal{H}_{M}$,
since each operator $H_{n}$ is semidefinite positive one has that
$\left\langle \phi\left|H_{n}\right|\phi\right\rangle =0$ iff $\ket{\phi}\in Ker\left(H_{n}\right)$.
Since we assume that the all $A_{n}$'s have non-degenerate eigenspectrum
one has that $\forall n$ $dim\left[Ker\left(H_{n}\right)\right]=M$
each $Ker\left(H_{n}\right)$ can be written as
\begin{eqnarray}
Ker\left(H_{n}\right) & = & span\left\{ \ket{a_{n,1}}\ket{a_{n,1}},\ket{a_{n,2}}\ket{a_{n,2}},..\right.\nonumber \\
 &  & \left....,\ket{a_{n,M}}\ket{a_{n,M}}\right\} \nonumber \\
\label{eq: Ker Hn eigenbasis}
\end{eqnarray}
 a fact which is easily derived by looking at the form of the generic
$H_{n}$ (\ref{eq: Hn Definition}): the states $\left\{ \ket{a_{n,i}}\ket{a_{n,i}}\right\} _{i=1}^{M}$
are mutually orthogonal, are all eigenstates of $H_{n}$ with zero
eigenvalue and they form an orthonormal basis of $Ker\left(H_{n}\right)$.
The Hamiltonian $H_{Tot}$ has $\varepsilon_{gs}=0$ iff $\cap_{n}Ker\left(H_{n}\right)\neq\oslash$
such that $\ket{\varepsilon_{gs}}\in\cap_{n}Ker\left(H_{n}\right)$
i.e., iff the intersection of the kernels of the $H_{n}$ operators
is not void and the ground state is a common eigenvector of all the
$H_{n}$ with zero energy. In order to derive the general form of
$\ket{\varepsilon_{gs}}$ we start by supposing that $\cap_{n}Ker\left(H_{n}\right)\neq\oslash$
and that there exist $\ket{\varepsilon_{gs}}\in\cap_{n}Ker\left(H_{n}\right)$.
We then focus on on a specific $H_{n}$, say $H_{1}$; since by hypothesis
$\ket{\varepsilon_{gs}}\in Ker\left(H_{1}\right)$ we write the state
in terms of the eigenbasis (\ref{eq: Ker Hn eigenbasis}) of $Ker\left(H_{1}\right)$
\begin{eqnarray*}
\ket{\varepsilon_{gs}} & = & \sum_{i=1}^{M}\alpha_{1,i}\ket{a_{1,i}}\ket{a_{1,i}}
\end{eqnarray*}
Since $\forall i$ one can write $\alpha_{1,i}=\left|\alpha_{1,i}\right|e^{i\phi_{1,i}}$
and reabsorb the phase factors in the definitions of the eigenvectors,
e.g. $\ket{\tilde{a}_{1,i}}=e^{i\phi_{1,i}/2}\ket{a_{1,i}}$ such
that 
\begin{eqnarray*}
\ket{\varepsilon_{gs}} & = & \sum_{i=1}^{M}\left|\alpha_{1,i}\right|\ket{\tilde{a}_{1,i}}\ket{\tilde{a}_{1,i}}
\end{eqnarray*}
In this way the ground state is written in its Schmidt decomposition
in terms of the basis $\left\{ \ket{\tilde{a}_{1,i}}\ket{\tilde{a}_{1,i}}\right\} _{i=1}^{M}$.
Since $\ket{\varepsilon_{gs}}\in\cap_{n}Ker\left(H_{n}\right)$ and
due to the structure (\ref{eq: Ker Hn eigenbasis}) of each $Ker\left(H_{n}\right)$,
the same is true for all $H_{n}$ such that one has
\begin{eqnarray}
\ket{\varepsilon_{gs}}=\sum_{i=1}^{M}\left|\alpha_{1,i}\right|\ket{\tilde{a}_{1,i}}\ket{\tilde{a}_{1,i}} & = & \sum_{i=1}^{M}\left|\alpha_{2,i}\right|\ket{\tilde{a}_{2,i}}\ket{\tilde{a}_{2,i}}=\nonumber \\
= & ... & =\sum_{i=1}^{M}\left|\alpha_{N,i}\right|\ket{\tilde{a}_{N,i}}\ket{\tilde{a}_{N,i}}\nonumber \\
\label{eq: GS HTot in Kernel basis}
\end{eqnarray}
 This result tells us that the ground state must be unique and that
$\forall i,n$ it must be $\left|\alpha_{n,i}\right|=1/\sqrt{M}$.
Indeed, each decomposition of the ground state (\ref{eq: GS HTot in Kernel basis})
represents in principle a \textit{different inequivalent versions}
of the Schmidt decomposition of $\ket{\varepsilon_{gs}}$. But for
a pure bipartite state, if the Schmidt coefficients $\left|\alpha_{n,i}\right|$
are not all degenerate i.e., all equal, than the Schmidt decomposition
is unique up to phase factors \cite{wernerappendix}. Since by hypothesis
$\ket{\varepsilon_{gs}}\in\cap_{n}Ker\left(H_{n}\right)$, in order
for the relation (\ref{eq: GS HTot in Kernel basis}) to be true,
in the first place it must be $\left|\alpha_{n,i}\right|=1/\sqrt{M},\ \forall n,i$.
Therefore if there is a common ground state this must read 
\begin{eqnarray}
\ket{\varepsilon_{gs}}=\frac{1}{\sqrt{M}}\sum_{i=1}^{M}\ket{\tilde{a}_{1,i}}\ket{\tilde{a}_{1,i}} & = & \frac{1}{\sqrt{M}}\sum_{i=1}^{M}\ket{\tilde{a}_{2,i}}\ket{\tilde{a}_{2,i}}=\nonumber \\
= & ... & =\frac{1}{\sqrt{M}}\sum_{i=1}^{M}\ket{\tilde{a}_{N,i}}\ket{\tilde{a}_{N,i}}\nonumber \\
\label{eq: General GS HTot in Kernel basis}
\end{eqnarray}
Now depending on the problem, there may or may not be the possibility
of adjusting the phases $\phi_{i,n}$ in order to have a single ground
state with $\varepsilon_{gs}=0$. In the affirmative case the ground
state of $H_{Tot}$ is unique and it can be written by using the appropriate
phases as $\ket{\varepsilon_{gs}}=\frac{1}{\sqrt{M}}\sum_{i}\ket{\tilde{a}_{n,i}}\ket{\tilde{a}_{n,i}},\ \forall n$.
Form which follows the first part of result $ii)$. It is actually
not important for the next part of the result to determine exactly
the various $\phi_{i,n}$. Indeed, the non-zero state-independent
lower bound $\varepsilon_{1}\left(1-\frac{1}{M}\right)$ can be derived
as follows. If $\varepsilon_{gs}=0$, given the general form of the
ground state derived above (\ref{eq: General GS HTot in Kernel basis})
i.e., that of a maximally entangled one, for any given $\ket{\phi}\in\mathcal{H}_{M}$
one can write 
\begin{eqnarray*}
\ket{\varepsilon_{gs}} & = & \frac{1}{\sqrt{M}}\sum_{i=1}^{M}\ket{a_{n,i}}\ket{a_{n,i}}\\
 & = & \frac{1}{\sqrt{M}}\left(\sum_{i=1}^{M}\ket{\phi_{n,i}}\ket{\phi_{n,i}^{*}}\right)
\end{eqnarray*}
where $\left\{ \ket{\phi_{n,i}}\right\} _{i=1}^{M}$ being mutually
orthonormal and $\ket{\phi}=\ket{\phi_{n,1}}$, while $\forall i\ $$\ket{\phi_{n,i}^{*}}$
is the complex conjugate of $\ket{\phi_{n,i}}$ when the latter is
expressed in the $\left\{ \ket{a_{n,i}}\right\} $ basis. The latest
formula allows to infer that $\max_{\ket{\phi}\in\mathcal{H}_{M}}\left|\bra{\phi}\bra{\phi}\left.\varepsilon_{gs}\right\rangle \right|^{2}=\max_{\ket{\phi}\in\mathcal{H}_{M}}\left|\bra{\phi}\left.\phi^{*}\right\rangle \right|^{2}/M=1/M$;
the maximum being attained by any state $\ket{\phi}=\sum_{i}U_{ji}\ket{a_{n,i}}$
with $U_{ji}\in\mathbb{R}$. Then, if $\left\{ \ket{\varepsilon_{n}}\right\} _{n=0}^{M^{2}-1}$
are the eigenstates of $H_{Tot}$ corresponding to the eigen-energies
$\varepsilon_{0}=\varepsilon_{gs}=0$ and $\varepsilon_{n}>0,\ \forall n=1,..,M^{2}-1$,
one has that $\forall\ket{\phi}\in\mathcal{H}_{M}$ 
\begin{eqnarray*}
\bra{\phi}\bra{\phi}H_{Tot}\ket{\phi}\ket{\phi} & = & \bra{\phi}\bra{\phi}\sum_{n=0}^{M^{2}-1}\varepsilon_{n}\ket{\varepsilon_{n}}\bra{\varepsilon_{n}}\ket{\phi}\ket{\phi}=\\
 & \ge & \varepsilon_{1}\sum_{n=1}^{M^{2}-1}\left|\bra{\phi}\bra{\phi}\left.\varepsilon_{n}\right\rangle \right|^{2}=\\
 & = & \varepsilon_{1}\bra{\phi}\bra{\phi}\left(\mathbb{I}_{M^{2}}-\ket{\varepsilon_{gs}}\bra{\varepsilon_{gs}}\right)\ket{\phi}\ket{\phi}=\\
 & = & \varepsilon_{1}\left(1-\left|\bra{\phi}\bra{\phi}\left.\varepsilon_{gs}\right\rangle \right|^{2}\right)
\end{eqnarray*}
Since 
\begin{eqnarray*}
\min_{\ket{\phi}\in\mathcal{H}_{M}}\varepsilon_{1}\left(1-\left|\bra{\phi}\bra{\phi}\left.\varepsilon_{gs}\right\rangle \right|^{2}\right) & = & \varepsilon_{1}\left(1-\frac{1}{M}\right)
\end{eqnarray*}
one has that $\forall\ket{\phi}\in\mathcal{H}_{M}$
\begin{eqnarray*}
V_{Tot}\left(\ket{\phi}\right)=\bra{\phi}\bra{\phi}H_{Tot}\ket{\phi}\ket{\phi} & \ge & \varepsilon_{1}\left(1-\frac{1}{M}\right)>0
\end{eqnarray*}
which is the second part of result $ii)$.

\section{Proof of proposition 4\label{Appendix: Proposition 4}}

We now prove the results of Proposition \ref{Prop 4}. We begin with
$i)$. Suppose $\alpha>0$, the proof is based on the analysis of
the Hamiltonian 
\begin{eqnarray*}
H_{Tot,n}^{\alpha} & = & \sum_{m\neq n}H_{m}+\frac{\left(A_{n}^{\alpha}\right)^{2}\otimes\mathbb{I}+\mathbb{I}\otimes\left(A_{n}^{\alpha}\right)^{2}}{2}=\\
 & = & H_{Tot,n}-\alpha\left(A_{n}\otimes\mathbb{I}+\mathbb{I}\otimes A_{n}\right)+\alpha^{2}\mathbb{I}
\end{eqnarray*}
where $H_{Tot,n}=\sum_{m\neq n}H_{m}+\frac{A_{n}^{2}\otimes\mathbb{I}+\mathbb{I}\otimes A_{n}^{2}}{2}$
is defined as above. If $\ket{\varepsilon_{gs,n}^{\alpha}}$ is a
ground state of $H_{Tot,n}^{\alpha}$ then $\ket{\varepsilon_{gs,n}^{-\alpha}}=U\otimes U\ket{\varepsilon_{gs,n}^{\alpha}}$
must be a ground state of $H_{Tot,n}^{-\alpha}$. Indeed, on one hand,
due to the symmetry properies of $\sum_{m\neq n}H_{m}$ that extend
to $H_{Tot,n}$, it holds $\bra{\varepsilon_{gs,n}^{-\alpha}}H_{Tot,n}\ket{\varepsilon_{gs,n}^{-\alpha}}=\bra{\varepsilon_{gs,n}^{\alpha}}H_{Tot,n}\ket{\varepsilon_{gs,n}^{\alpha}}$.
Furthermore, due to the action of $U$ on $A_{n}$ 
\begin{eqnarray*}
\bra{\varepsilon_{gs,n}^{-\alpha}}\left(A_{n}\otimes\mathbb{I}+\mathbb{I}\otimes A_{n}\right)\ket{\varepsilon_{gs,n}^{-\alpha}} & =\\
=-\bra{\varepsilon_{gs,n}^{\alpha}}\left(A_{n}\otimes\mathbb{I}+\mathbb{I}\otimes A_{n}\right)\ket{\varepsilon_{gs,n}^{\alpha}}
\end{eqnarray*}
such that 
\begin{eqnarray*}
\varepsilon_{gs,n}^{-\alpha}=\bra{\varepsilon_{gs,n}^{-\alpha}}H_{Tot,n}^{-\alpha}\ket{\varepsilon_{gs,n}^{-\alpha}} & = & \bra{\varepsilon_{gs,n}^{\alpha}}H_{Tot,n}^{\alpha}\ket{\varepsilon_{gs,n}^{\alpha}}=\varepsilon_{gs,n}^{\alpha}
\end{eqnarray*}
Then $ii)$ simply follows from the fact that 
\begin{eqnarray*}
\bra{\varepsilon_{gs,n}^{0}}\left(A_{n}\otimes\mathbb{I}+\mathbb{I}\otimes A_{n}\right)\ket{\varepsilon_{gs,n}^{0}} & =\\
-\bra{\varepsilon_{gs,n}^{0}}\left(A_{n}\otimes\mathbb{I}+\mathbb{I}\otimes A_{n}\right)\ket{\varepsilon_{gs,n}^{0}}
\end{eqnarray*}
and there for to first order in $\delta\alpha\ll1$ one has $\varepsilon_{gs,n}^{\delta\alpha}=\varepsilon_{gs,n}+\delta\alpha^{2}\ge\varepsilon_{gs,n}^{0}$.

\section{Symmetries for spin hamiltonian\label{sec: Appendix: symmetry for spin Ham}}

In this Appendix we detail the symmetries property of $H_{Tot}$ (\ref{eq: HTot JX JZ})
defined in terms of the two spin operators $J_{X},J_{Z}$. One has
that 
\begin{eqnarray*}
e^{-i\theta J_{Y}}J_{Z}e^{i\theta J_{Y}} & = & \cos\theta J_{Z}+\sin\theta J_{X}\\
e^{-i\theta J_{Y}}J_{X}e^{i\theta J_{Y}} & = & -\sin\theta J_{Z}+\cos\theta J_{X}
\end{eqnarray*}
then, given $U_{YY}=e^{-i\theta J_{Y}}\otimes e^{-i\theta J_{Y}}$
\begin{eqnarray*}
U_{YY}J_{Z}\otimes J_{Z}U_{YY}^{\dagger} & = & \cos^{2}\theta J_{Z}\otimes J_{Z}+\sin^{2}\theta J_{X}\otimes J_{X}+\\
 & + & \sin\theta\cos\theta\left(J_{Z}\otimes J_{X}+J_{X}\otimes J_{Z}\right)\\
U_{YY}J_{X}\otimes J_{X}U_{YY}^{\dagger} & = & \sin^{2}\theta J_{Z}\otimes J_{Z}+\cos^{2}\theta J_{X}\otimes J_{X}+\\
 & - & \sin\theta\cos\theta\left(J_{Z}\otimes J_{X}+J_{X}\otimes J_{Z}\right)
\end{eqnarray*}
such that 
\begin{eqnarray*}
U_{YY}\left(J_{Z}\otimes J_{Z}+J_{X}\otimes J_{X}\right)U_{YY}^{\dagger} & = & \left(J_{Z}\otimes J_{Z}+J_{X}\otimes J_{X}\right)
\end{eqnarray*}
Furthermore by using the Casimir relation $j(j+1)\mathbb{I=}J_{X}^{2}+J_{Y}^{2}+J_{Z}^{2}$
the Hamiltonian $H_{Tot}$ can be expressed as 
\begin{eqnarray*}
H_{Tot} & = & \frac{\left(J_{Z}^{2}+J_{X}^{2}\right)\otimes\mathbb{I}+\mathbb{I}\otimes\left(J_{Z}^{2}+J_{X}^{2}\right)}{2}+\\
 & - & \left(J_{Z}\otimes J_{Z}+J_{X}\otimes J_{X}\right)=\\
 & = & j(j+1)\mathbb{I}\otimes\mathbb{I}-\frac{J_{Y}^{2}\otimes\mathbb{I}+\mathbb{I}\otimes J_{Y}^{2}}{2}+\\
 & - & \left(J_{Z}\otimes J_{Z}+J_{X}\otimes J_{X}\right)
\end{eqnarray*}
such that 
\begin{eqnarray*}
U_{YY}H_{Tot}U_{YY}^{\dagger} & = & H_{Tot}
\end{eqnarray*}
therefore $\forall\ket{\phi}\in\mathcal{H}_{M}$ if 
\begin{eqnarray*}
\bra{\phi}\bra{\phi}H_{Tot}\ket{\phi}\ket{\phi} & = & c(\phi)
\end{eqnarray*}
then one has also that 
\begin{eqnarray*}
\bra{\phi}\bra{\phi}H_{Tot}\ket{\phi}\ket{\phi} & = & \bra{\phi}\bra{\phi}U_{YY}H_{Tot}U_{YY}^{\dagger}\ket{\phi}\ket{\phi}=\\
 & = & \bra{\phi_{\theta}}\bra{\phi_{\theta}}H_{Tot}\ket{\phi_{\theta}}\ket{\phi_{\theta}}\\
 & = & c(\phi)
\end{eqnarray*}
Therefore one has a certain degrees of freedom in choosing $\ket{\phi}$
since all states $\ket{\phi_{\theta}}=e^{i\theta J_{Y}}\ket{\phi},\forall\theta\in\mathbb{R}$
will have the same variance $c(\phi)$ . Now 
\begin{eqnarray*}
\bra{\phi_{\theta}}J_{x}\ket{\phi_{\theta}} & = & -\sin\theta\bra{\phi}J_{z}\ket{\phi}+\cos\theta\bra{\phi}J_{x}\ket{\phi}
\end{eqnarray*}
Suppose now $\ket{\phi}$ is a state which minimizes $V_{XZ}$. One
can always choose for example $\theta$ such that 
\begin{eqnarray*}
\bra{\phi_{\theta}}J_{x}\ket{\phi_{\theta}} & = & 0
\end{eqnarray*}
i.e., we can choose $\theta$ by setting 
\begin{eqnarray*}
\sin\theta\bra{\phi}J_{z}\ket{\phi} & = & +\cos\theta\bra{\phi}J_{x}\ket{\phi}\\
\tan\theta & = & \frac{\bra{\phi}J_{x}\ket{\phi}}{\bra{\phi}J_{z}\ket{\phi}}\\
\theta & = & \arctan\left(\frac{\bra{\phi}J_{x}\ket{\phi}}{\bra{\phi}J_{z}\ket{\phi}}\right)
\end{eqnarray*}
Therefore even if $\theta$ is unknown we can find the lower bound
of $V_{XZ}$ by finding the ground state of the Hamiltonian 
\begin{eqnarray*}
H_{Tot,X} & = & \frac{\left(J_{Z}^{2}+J_{X}^{2}\right)\otimes\mathbb{I}+\mathbb{I}\otimes\left(J_{Z}^{2}+J_{X}^{2}\right)}{2}-J_{Z}\otimes J_{Z}
\end{eqnarray*}
Indeed $\varepsilon_{Tot,X}^{0}$ will give a lower bound $\forall\ket{\phi}\in S_{X}^{0}$
among which there will be the $\ket{\phi_{\theta}}$ which minimizes
$V_{XZ}$. Then $\forall\ket{\psi}\in\mathcal{H}_{M}$ one has 
\begin{eqnarray*}
V_{XZ}(\ket{\psi}) & \ge & V_{XZ}(\ket{\phi_{\theta}})=\\
 & \ge & \varepsilon_{gs,X}^{0}
\end{eqnarray*}

\section{Planar spin squeezing\label{Appendix: Planar Spin Squeezing}}

In this Appendix we show how from the knowledge of $\ket{\lambda_{Max}^{+}},\ket{\lambda_{Max}^{-}}$
one can obtain a state $\ket{\psi_{sat}}=\ket{\theta_{m}}$ that can
in principle realized in the laboratory and that approximately saturates
the bound for planar spin squeezing. For fixed $j$ one can study
the profile of $\ket{\lambda_{Max}^{+}},\ket{\lambda_{Max}^{-}}$;
a feature that holds for all analyzed values of $j$ is that the profile
is peaked at $m_{z}=j$ and $m_{z}=-j$ respectively, and such that
only the states with $m_{z}=-j+2k$ have non-zero amplitudes. These
numerical findings will lead us in the search for states $\ket{\psi_{sat}}$
that on one hand are a good approximations of $\ket{\lambda_{Max}^{+}},\ket{\lambda_{Max}^{-}}$
and on the other hand are in principle obtainable in the laboratory.\\
We start by considering the relation (\ref{eq: VXZ}) which, over
the set of eigenstates of $J_{Z}$, is minimized by $\ket{j,\pm j}$
and for such states $\Delta^{2}J_{Z}=0$ and $V_{XZ}=\Delta^{2}J_{X}=j/2$.
In order to obtain a lower bound for $V_{XZ}$ smaller than $j/2$,
one can imagine to start from the state $\ket{j,j}$ for example and
to modify it in such a way that $\Delta^{2}J_{Z}\gtrapprox0$ is little
changed and at the same time $\Delta^{2}J_{X}$ is considerably reduced.
This heuristic reasoning suggests the strategy of searching for an
operator $G$ such that $\ket{\theta}=\exp\left(-i\theta G\right)\ket{j,j}_{Z}$
$\theta\in\mathbb{R}$ is the state required. If one analyses $V_{XZ}^{\theta}=V_{XZ}(\ket{\theta})$
and in particular its first order variation $\partial_{\theta}V_{XZ}^{\theta}$
in $\theta=0$ one has 
\begin{eqnarray*}
\partial_{\theta}\left[\Delta^{2}J_{Z}\left(\theta\right)\right]_{\theta=0} & = & 0\\
\partial_{\theta}\left[\Delta^{2}J_{X}\left(\theta\right)\right]_{\theta=0} & = & \left\langle j,j\left|\left[J_{X}^{2},G\right]\right|j,j\right\rangle +\\
 & - & \left\langle j,j\left|J_{X}\right|j,j\right\rangle \left\langle j,j\left|\left[J_{x},G\right]\right|j,j\right\rangle =\\
 & = & 2\ Im\left[\left\langle j,j-2\left|G\right|j,j\right\rangle \right]
\end{eqnarray*}
The previous relations thus leads to consider operators for which
$\bra{j,j-2}G\ket{j,j}_{z}\neq0$. The above reasoning heuristically
leads to analyze the action of the two-axis squeezing operator 
\begin{eqnarray*}
H_{TAS} & = & -i\left(J_{+}^{2}-J_{-}^{2}\right)
\end{eqnarray*}
which is known to have the property of squeezing along the $X$ axis
and simultaneously anti-squeezed along the $Y$ axis. This latter
property is consistent with the relation (\ref{eq: V_XYZ}) where
it can be seen that any attempt to squeeze the sum $V_{XZ}$ implies
the enhancement of $\Delta^{2}J_{Y}$. The action of the operator
$U=\exp\left(-i\theta H_{TAS}\right)$ on $\ket{j,j}$ is not known
in an analytical form, however it has the desirable property of populating
only the basis states $\ket{j,j-2k}$ thus reproducing one of the
features of the states $\ket{\lambda_{Max}^{+}},\ket{\lambda_{Max}^{-}}$
discussed above. \\Following the previous discussion the goal now
is to find the optimal value $\theta_{m}$ of the squeezing parameter
$\theta$ such that the state $\ket{\psi_{sat}}=\ket{\theta_{m}}=\exp\left(-i\theta_{m}H_{TAS}\right)\ket{j,j}_{Z}$
approximately saturates the lower bound for $V_{XZ}$. This in principle
requires for each $j$ the numerical search for the optimal value
of $\theta_{m}=\theta_{m}(j)$ for which the minimum of $V_{XZ}^{\theta}$
is attained. We now show how to analytically estimate the optimal
value of $\theta_{m}$. As anticipated in the main text we resort
to the Holstein-Primakoff (HP) transformation that allows to map the
spin operators to harmonic oscillators ones. Indeed as shown in \cite{HolsteinPrimakoff,EmaryDickeHP,WerAngularMomentum}
one can write the spin operators in terms of the bosonic creation
and annihilation operators $a,a^{\dagger}$. 
\begin{eqnarray*}
J_{+} & = & \sqrt{2j}a^{\dagger}\sqrt{1-\frac{a^{\dagger}a}{2j}}\\
J_{.-} & = & \sqrt{2j}\sqrt{1-\frac{a^{\dagger}a}{2j}}a\\
J_{z} & = & a^{\dagger}a-j
\end{eqnarray*}
such that for states with average number of bosons $\left\langle \hat{n}\right\rangle =\left\langle a^{\dagger}a\right\rangle \ll2j$
one has that $J_{+}=\sqrt{2j}a^{\dagger},\ J_{-}=\sqrt{2j}a$. With
this transformation the sum of variances (\ref{eq: VXZ}) can be written
as 
\begin{eqnarray}
V_{XZ}^{bos} & = & \Delta^{2}\hat{n}+j\Delta^{2}\hat{x}\label{eq: VXZBoson}
\end{eqnarray}
where: $\hat{n}$ is the number operator; $ $ $\hat{x}=(a+a^{\dagger})/\sqrt{2}$
is the position operator and $\Delta^{2}J_{Z}\rightarrow\Delta^{2}\hat{n}$$\Delta^{2}J_{X}\rightarrow j\Delta^{2}\hat{x}$.
Within the Holstein Primakoff representation the spin state $\ket{j,j}$
is mapped into the vacuum $\ket{0}$. In general there is no such
mapping between the squeezed state $\ket{\theta}$ and the corresponding
single mode squeezed vacuum state that reads \cite{ParisGaussianStatesInQInfo}
\begin{eqnarray*}
\ket{\xi} & = & \exp\left\{ \frac{1}{2}\left[\xi\left(a^{\dagger}\right)^{2}-\xi^{*}a^{2}\right]\right\} \ket{0}
\end{eqnarray*}
with $\xi=re^{-i\phi}$ the squeezing parameter. However, this state
is the ``natural'' counterpart of $\ket{\theta}$ in the search
for a minimum of $V_{XZ}^{bos}$ . Within the HP framework two-axis
squeezing operator transforms into the single-mode squeezing operator
\begin{eqnarray*}
e^{-i\theta H_{TAS}} & = & \exp\left[-\theta\left(J_{+}^{2}-J_{-}^{2}\right)\right]=\\
 & \approx & \exp\left\{ -\theta2j\left[\left(a^{\dagger}\right)^{2}-a^{2}\right]\right\} 
\end{eqnarray*}
such that if we now choose $\xi=-4j\theta$ we can bridge the spin
and the bosonic version of $V_{XZ}$. With these assumptions $V_{XZ}^{bos}$
reads 
\begin{eqnarray}
V_{XZ}^{boson}\left(\theta\right) & = & 2\sinh^{2}(4j\theta)\cosh^{2}(4j\theta)+j\frac{\exp\left(8j\theta\right)}{2}\nonumber \\
\label{eq: V_XZ bosons}
\end{eqnarray}
The minimization of the latter expression with respect to $\theta$
provides a single real solution that for $j\gg1$ can be written as
\begin{eqnarray}
\theta_{m} & = & -\frac{\log2+\log j}{24\ j}+o\left(1/j^{2}\right)\label{eq: theta min HP-1}
\end{eqnarray}
such that for $j\gg1$ one finds 
\begin{eqnarray*}
V_{XZ}^{boson}\left(\theta_{m}\right) & \approx & 0.595275\ j^{2/3}
\end{eqnarray*}
We notice that the scaling obtained in the HP framework \textit{coincides}
with the dominant part of (\ref{eq: VXZ He scaling}) for large $j$.
The found approximate solution $\theta_{m}$ can now be used to compute
the bound for the spin version of the sum of variances (\ref{eq: VXZ})
i.e., $V_{XZ}(\ket{\theta_{m}})$. The consequences of this results
are described in the Main text. 

\section{The bosonic case: gaussian states\label{sec: Appendix: bosonic-case xn}}

The generic pure Gaussian state reads 
\begin{eqnarray*}
D(\alpha)S(\xi)\ket{0} & = & \ket{\alpha,\xi}
\end{eqnarray*}
The variance of $x$ for such states can thus be written as 
\begin{eqnarray*}
\Delta_{\ket{\alpha,\xi}}^{2}x & = & \left\langle \alpha,\xi\left|x^{2}\right|\alpha,\xi\right\rangle -\left\langle \alpha,\xi\left|x\right|\alpha,\xi\right\rangle ^{2}=\\
 & = & \left\langle \xi\left|D^{\dagger}(\alpha)xD(\alpha)D^{\dagger}(\alpha)xD(\alpha)\right|\xi\right\rangle +\\
 & - & \left\langle \xi\left|D^{\dagger}(\alpha)xD(\alpha)\right|\xi\right\rangle ^{2}=\\
 & = & \Delta_{\ket{\xi}}^{2}x_{\alpha}
\end{eqnarray*}
with $x_{\alpha}=D^{\dagger}(\alpha)xD(\alpha)=x+2Re\left[\alpha\right]\mathbb{I}$.
Since $\Delta^{2}\left[A+c\mathbb{I}\right]=\Delta^{2}A$ one has
that $\Delta_{\ket{\alpha,\xi}}^{2}\hat{x}=\Delta_{\ket{\xi}}^{2}\hat{x}$.
i.e., the displacement does not change the variance of $x$, since
it only changes its average value. We now evaluate the variance of
$\hat{n}$ and find $\Delta_{\ket{\alpha,\xi}}^{2}\hat{n}=\Delta_{\ket{\xi}}^{2}\hat{n}_{\alpha}$
with $n_{\alpha}=n+a^{\dagger}\alpha+a\alpha^{*}+\left|\alpha\right|^{2}$.
The constant $\left|\alpha\right|^{2}$ again can be dropped and one
is left with such that
\begin{eqnarray*}
\Delta_{\ket{\xi}}^{2}\hat{n}_{\alpha} & = & \Delta_{\ket{\xi}}^{2}\hat{n}+2\left|\alpha\right|^{2}\Delta_{\ket{\xi}}^{2}\hat{x}_{\arg\alpha}+\\
 & + & \left|\alpha\right|\left[\left\langle \hat{n}\hat{x}_{\arg\alpha}\right\rangle +\left\langle \hat{x}_{\arg\alpha}\hat{n}\right\rangle -2\left\langle \hat{n}\right\rangle \left\langle \hat{x}_{\arg\alpha}\right\rangle \right]
\end{eqnarray*}
where $\hat{x}_{\arg\alpha}=(ae^{i\ \arg\alpha}+a^{\dagger}e^{-i\ \arg\alpha})/\sqrt{2}$.
Since the averages are taken for the state $\ket{\xi}$, for the property
of the latter one has $\left\langle \hat{n}\hat{x}_{\arg\alpha}\right\rangle =\left\langle \hat{x}_{\arg\alpha}\hat{n}\right\rangle =\left\langle \hat{x}_{\arg\alpha}\right\rangle =0$.
Overall the previous results show that, $\forall\alpha,\xi$ i.e.,
for all pure Gaussian states $\ket{\alpha,\xi}$ 
\begin{eqnarray*}
\Delta_{\ket{\alpha,\xi}}^{2}n+\Delta_{\ket{\alpha,\xi}}^{2}x & = & \Delta_{\ket{\xi}}^{2}n+2\left|\alpha\right|^{2}\Delta_{\ket{\xi}}^{2}x_{\arg\alpha}+\Delta_{\ket{\xi}}^{2}x\ge\\
 & \ge & \Delta_{\ket{\xi}}^{2}n+\Delta_{\ket{\xi}}^{2}x
\end{eqnarray*}
such that the minimum of $V_{xn}$ over the set of Gaussian state
is given by the squeezed vacuum state $\ket{\xi_{m}}$ that minimizes
$\Delta_{\ket{\xi}}^{2}n+\Delta_{\ket{\xi}}^{2}x$. 
\end{document}